\documentclass[11pt,a4wide]{article}
\pdfoutput=1
\usepackage{amssymb,a4wide}
\usepackage{amsmath}
\usepackage{latexsym,url}
\usepackage{ascii}
\usepackage{sgame}
\usepackage{color}
\usepackage{xcolor}
\usepackage{xspace}
\usepackage{ntheorem}
\usepackage{dsfont}
\usepackage{extarrows}
\usepackage{authblk}

\usepackage{tikz}
\usetikzlibrary{automata}
\usepackage{pgf,pgfarrows,pgfnodes}
\usepackage{multicol,xspace}
\usepackage {pgflibraryshapes}
\usepackage{eurosym}
\usetikzlibrary{arrows,backgrounds,positioning,fit,calc,shadows}

\newtheorem{theorem}{Theorem}

\newtheorem{corollary}[theorem]{Corollary}

\newtheorem{definition}[theorem]{Definition}

\newtheorem{lemma}[theorem]{Lemma}

\newtheorem{proposition}[theorem]{Proposition}
\newtheorem{remark}[theorem]{Remark}

\newtheorem{example}[theorem]{Example}
\newenvironment{proof}[1][Proof]{\textbf{#1.} }{\ \rule{0.5em}{0.5em}}

  \newcommand{\cut}[1]{}
  \newcommand{\vcut}[1]{}
 \newcommand{\pcut}[1]{}
     
\newcommand{\vnote}[1]{{\ \color{red}\bf{}{VG: #1}}}

 \newcommand{\solc}{\ensuremath{\mathfrak{S}}\xspace}
 \newcommand{\sol}[1]{\mathfrak{S}(#1)}
 \newcommand{\val}{\mathsf{v}}

\newcommand{\turn}{\ensuremath{\mathsf{turn}}}

 \newcommand{\bfg}{\ensuremath{\mathbf{g}}}
 \newcommand{\bfu}{\ensuremath{\mathbf{u}}}
 
 \newcommand{\act}{\ensuremath{\mathcal{A}}}
  \newcommand{\hist}{\ensuremath{\mathcal{H}}}
  \newcommand{\fhist}{\ensuremath{\mathcal{H}^{f}}}  
  \newcommand{\trmhist}{\ensuremath{\mathcal{H}^{t}}}    
  \newcommand{\N}{\ensuremath{\mathrm{N}}}
  \newcommand{\out}{\mathsf{out}}

\begin{document}
\title{Non-cooperative games with preplay negotiations} 
\author[1]{Valentin Goranko}
\author[2]{Paolo Turrini}

\affil[1]{\small{Department of Applied Mathematics and Computer Science, Technical University of Denmark}} 
\affil[2]{\small{Department of Computing, Imperial College London}}
\maketitle

\begin{abstract}
We consider an extension of  strategic normal form games with a phase of negotiations before the actual play of the game, where players can make binding offers for transfer of utilities to other players after the play of the game, in order to provide additional incentives for each other to play designated strategies. The enforcement of such offers is conditional on the recipients playing the specified strategies and they effect transformations of the payoff matrix of the game by accordingly transferring payoffs between players. Players can exchange series of such offers in a preplay negotiation game in an extensive form. We introduce and analyze solution concepts for 
 normal form games with such preplay offers under various assumptions for the preplay negotiation phase and obtain results for existence of efficient negotiation strategies of the players. 
\end{abstract}

\section{Introduction}\label{sec:intro}

It is well known that some normal form games have no pure strategy Nash equilibria, while others, like the Prisoner's Dilemma, have rather unsatisfactory -- e.g., strongly Pareto dominated -- ones. These inefficiencies are often attributed to the lack of communication between the players and the impossibility for them to agree on a mutually beneficial joint course of action, before the actual play of the game. Indeed, undesirable outcomes could often be avoided if players were able to communicate and make binding agreements on the strategy to play before the game starts, by  signing contracts. However, even if players could freely communicate before the game, enforcing of such contracts is often not possible in practice and, furthermore, it would change the nature of the game from non-cooperative to essentially cooperative.

\medskip 

Here we consider a weaker and generally more realistic assumption, viz.: 

\begin{quote}
 \emph{Before the actual game is played any player, say $A$, can make a binding 
offer to any other player, say $B$, to pay him\footnote{We refer to player $A$ 
as a female, while to $B$ as a male. This choice is not for the sake of political correctness but to make it easier to distinguish the players from the context.}, after the game is played, an explicitly declared amount of utility $\delta$ if $B$ plays a strategy $s$ specified in the offer by $A$.}
\end{quote}

Building up upon this basic, {\em unconditional}, form of offer, we also consider a more general setting, where players, without acting as a coalition, can propose a game transformation to their fellow players, by making an offer under the condition of receiving another offer in return, proposal that can, in turn, be accepted or rejected. This newly obtained game transformation can be further modified, with proposals made by other players, until an agreement is reached. When endowing players with the possibility of playing such extra pre-play moves, a whole {\em bargaining phase} emerges before a normal form game is actually played. In other words, we can think of the normal form game that is eventually played as an outcome of another game, played beforehand, in which players engage in exchanging offers on strategies of other players until an agreement is reached on the game to play.  

Introducing an extensive-form bargaining structure preceding the play of a normal form game  is relevant  and important for the analysis of a wide spectrum of economic, social and political situations, such as:  
\vcut{
\emph{corruption schemes}, involving bribes in exchange of illegal favors; 
\emph{collusions} between two or more parties in an economic activity by exchanging `behind the curtain' agreements for mutual incentives; 
 \emph{kickback schemes} and other quasi-legal incentives;  
political, labour-related or business \emph{compromises} between non-cooperative parties, 
 \emph{compensations, concessions, out-of-court settlements of legal cases}, etc. 
}
\begin{itemize}
\itemsep = -2pt
\item \emph{corruption schemes} involving bribes in exchange of illegal favors; 
\item \emph{collusions} between two or more parties in an economic activity, by exchanging `behind the curtain' agreements for mutual incentives.
\item \emph{kickback schemes} and other quasi-legal incentives, 
\item political, labour-related or business \emph{negotiations and compromises} between 
non-cooperative parties, 
\item \emph{compensations, concessions, out-of-court settlements of legal cases}, etc.
\end{itemize}
For further details and discussions of these kinds of scenarios see for instance in \cite{gut1, gut2, schelling, segal}.

We note that agreements in such economic and political negotiations are usually reached in dynamic bargaining processes made of offers and counteroffers, rather than a one-shot simultaneous proposal ending the talks.

The literature in economic theory abounds with examples of parties entering negotiations to overcome inefficient resource allocation, as well as schemes of side payments, compensatory mechanisms, etc., which we review in detail in Section \ref{sec:related}. Here we only mention 
some more recent studies of pre-play contracting in games that consider one-shot simultaneous, in \cite{JW05}, \cite{EP11}, or two-step, in \cite{yamada-contracts}, offers preceding the actual game play and conditional on the entire strategy profile (see discussion in Section \ref{sec:related}). Somewhat surprisingly, however, a systematic study of the extensive-form negotiation process preceding the actual game play seems still to be missing in the literature. 
With this paper we initiate such systematic study purporting to fill this gap, by formalizing and studying the negotiation process preceding the actual game play as a bargaining among the players on the game to play, thus drawing connections with modern bargaining theory, in particular, Rubinstein's model of bargaining games \cite{OR90,OR}.
The paper is intended as a research `manifesto' in which we introduce and discuss conceptually our framework and outline a long term research agenda on it. In particular, we discuss our framework in more detail in Section \ref{sec:fram}, illustrate and discuss preplay offers and offer-induced game transformations in Section \ref{sec:trans} and introduce normal form games with preplay negotiations phase in Section \ref
{sec:preplay}. Then we analyze the case of 2-player normal form games with preplay offers with unconditional offers under various assumptions for the preplay negotiation phase in Section \ref{sec:2-player-uncond} and then we analyze the case with conditional offers in Section \ref{sec:2-player-cond}, where we obtain results for existence of efficient negotiation strategies of both players, significantly extending our work in \cite{GT-LORI2013}. 
 We end the paper with discussion of related work in Section \ref{sec:related}
 and concluding remarks and directions for further study  in Section \ref{sec:conc}. 


\section{Non-cooperative games with preplay offers: \\ the conceptual framework} 
\label{sec:fram}

In this section we provide a more detailed description of preplay offers, discuss some motivating examples, and lay down several extra conditions that play a role in determining the outcome of the negotiation phase.

\subsection{Preplay offers in more detail}
 We assume that any preplay offer by $A$ to $B$ is \emph{binding for $A$}, conditional on $B$ playing the strategy $s$ specified by $A$\footnote{We will not discuss here the mechanism securing the payments of the preplay offers after the play if the conditions are met. That can be done by a legal contract, or by using a trusted third party, etc.}. 
 However, such offer \emph{does not} create any obligation for $B$ and therefore it does not transform the game into a cooperative one, for $B$ is still at liberty to choose his strategy when the game is actually played. In particular, after her offer $A$ does not know before the game is played whether $B$ will play the desired by $A$ strategy $s$, and will thus make use of the offer, or not. Furthermore, several such offers can be made, possibly by different players, so the possible rational behaviours of the payers game maintain, in principle, all their complexity.  The key observation applying to this assumption, is that \emph{after any binding preplay offer is made, the game remains a standard non-cooperative normal form game, only the payoff matrix changes according to the offer.}


\subsection{Motivating examples}\label{sec:motivating}

First, we introduce the following notation: $A \xlongrightarrow{\delta \slash  \sigma_B} B$ denotes an offer made by player $A$ to pay an amount $\delta$ to player $B$ after the play of the game if player $B$ plays strategy $\sigma_B$.  

\paragraph{Prisoners' Dilemma 1}
Consider a standard version of the Prisoner' s Dilemma (PD) game in Figure \ref{PD}.
\begin{figure}[htb]\hspace*{\fill}%
\begin{game}{2}{2}
      & $C$    & $D$\\
$C$   &$4,4$   &$0,5$\\
$D$   &$5,0$   &$1,1$
\end{game}\hspace*{\fill}%
\caption{Prisoner's Dilemma 1}
\label{PD}
\end{figure}
The only Nash Equilibrium (NE) of the game is $(D,D)$, yielding a payoff of $(1,1)$. Now, suppose 
$\mathit{Row} \xlongrightarrow{2 \slash  C} \mathit{Column}$, that is, 
player $\mathit{Row}$ makes to the player $\mathit{Column}$ 
a \emph{binding offer} to pay her 2 units of utility (hereafter, \emph{utils}) after the game if $\mathit{Column}$ plays $C$. 
That offer transforms the game by transferring 2 utils from the payoff of $\mathit{Row}$ to the payoff of $\mathit{Column}$ in every entry of the column where $\mathit{Column}$ plays $C$, as pictured in Figure \ref{PDoffer}.

\begin{figure}[htb]\hspace*{\fill}%
\begin{game}{2}{2}
      & $C$    & $D$\\
$C$   &$2,6$   &$0,5$\\
$D$   &$3,2$   &$1,1$
\end{game}\hspace*{\fill}%
\caption{An offer to cooperate by player Row.}
\label{PDoffer}
\end{figure}

In this game player $\mathit{Row}$ still has the incentive\footnote{Intuitively, {\em having the incentive} to play a strategy should be understood as realizing that that strategy is not dominated. Later on we will provide a formal and abstract notion of equilibrium, which will rule out dominated strategies to be part of the solution of a game.} to play $D$, which strictly dominates $C$ for him, but the dominant strategy for $\mathit{Column}$ now is $C$, and thus the only Nash equilibrium is $(D,C)$ with payoff $(3,2)$ -- strictly dominating the original payoff $(1,1)$.

Thus, even though player $\mathit{Row}$ will still defect, the offer he has made to player $\mathit{Column}$ makes it strictly better for $\mathit{Column}$ to cooperate.

Of course, $\mathit{Column}$ can now realize that if player $\mathit{Row}$ is to cooperate,
then $\mathit{Column}$ would be even better off, but for that an extra incentive for $\mathit{Row}$ is needed. That incentive can be created by an offer  $\mathit{Column} \xlongrightarrow{2 \slash  C} \mathit{Row}$, that is, 
if $\mathit{Column}$, too, makes an offer to $\mathit{Row}$ to pay him 2 utils after the game, if player $\mathit{Row}$ cooperates. Then the game transforms, as in Figure \ref{PDoffer2}.

\begin{figure}[htb]\hspace*{\fill}%
\begin{game}{2}{2}
      & $C$    & $D$\\
$C$   &$4,4$   &$2,3$\\
$D$   &$3,2$   &$1,1$
\end{game}\hspace*{\fill}%
\caption{A second offer, by player Column.}
\label{PDoffer2}
\end{figure}

In this game, the only Nash equilibrium is $(C,C)$ with payoff $(4,4)$, which is also Pareto optimal. Note that this is the same payoff for $(C,C)$ as in the original PD game, but now both players have created incentives for their opponents to cooperate, and have thus escaped from the trap of the original inefficient Nash equilibrium $(D,D)$.

\begin{remark}
Clearly, preplay offers can only work in case when at least part of the received payoff can actually be transferred from a player to another. They obviously cannot apply to scenarios such as the original PD, where one prisoner cannot offer to the other to stay in prison for him, even if they could communicate before the play. 
\end{remark}

\paragraph{Prisoners' Dilemma 2}
\label{example:PD2}
Consider another version of the Prisoner' s Dilemma game in Figure \ref{PD2}.
\begin{figure}[htb]\hspace*{\fill}%
\begin{game}{2}{2}
      & $C_{Col}$   & $D_{Col}$ \\ 
$C_{Row}$   &$4,4$   &$0,5$\\
$D_{Row}$   &$5,0$   &$3,3$
\end{game}\hspace*{\fill}%
\caption{Prisoner's Dilemma 2}
\label{PD2}
\end{figure}
The only Nash Equilibrium in this game is  $(D_{Row},D_{Col})$, yielding the Pareto dominated payoff of $(3,3)$. Now, note that none of the players can make a feasible first offer to improve the outcome. Indeed, in order to provide a sufficient incentive for $\mathit{Column}$ to play $C_{Col}$,  
$\mathit{Row}$ would have to offer him more than 3, which is unfeasible for $\mathit{Row}$ because it would put him in a disadvantaged position. Likewise for $\mathit{Column}$. 
 
Thus, by consecutive exchange of unilateral preplay offers rational players cannot realize the opportunity to play the Pareto optimal outcome 
$(C_{Row},C_{Col})$. 

This problem can be avoided if we allow \emph{conditional offers} as follows: $\mathit{Row}$  can make an offer $\mathit{Row} \xlongrightarrow{3 \slash  C_{Col}} \mathit{Column}$, but now, \emph{conditional on 
$\mathit{Column}$ making to $\mathit{Row}$  the matching counter-offer}  
$\mathit{Column} \xlongrightarrow{3 \slash  C_{Row}} \mathit{Row}$, which we hereafter denote as $\mathit{Row} \xlongrightarrow{3 \slash  C_{Col} \ \mid \ 3 \slash  C_{Row}} \mathit{Row}$. The idea is that, unlike the so far considered unconditional offers, 
$\mathit{Row}$ 's conditional offer is only confirmed and enforced if 
$\mathit{Row}$ does make the required counter-offer, else it is cancelled and nullified before the play of the game.

We will introduce formally and discuss conditional offers in detail further. 

\subsection{Additional optional assumptions}

There are several important additional assumptions that, depending on the particular scenarios under investigation may, or may not, be realistically made. We therefore do not commit to any of them, but we acknowledge that each of them can make a significant difference in the behaviour and abilities of players to steer the game in the best possible direction for them. So, we consider the possible options for each of them separately and study the consequences under the various combinations of assumptions.

\begin{description}
\item[Revocability of offers.]
Once made, offers may, or may not, be withdrawn during the negotiations phase. Both cases are reasonable and realistic, and we consider each of them separately.

\item[Value of time.] Time, measured discretely as the number of explicitly defined steps/rounds of the negotiations,  may or may not have value, i.e. players may, or may not, strictly prefer a reward in the present to the same reward in the future. Moreover, time may have the same value for all players, or may be more, or less, valuable for each of them depending on their {\em patience}.  
\begin{itemize}
\itemsep = -2pt
\item In the case when time is of no value, players can keep making and withdrawing offers (if allowed to do so) at no extra cost. Intuitively the effect should be the same as if withdrawn offers were never made. 

\item In the case when time is of value, making unacceptable or suboptimal offers or withdrawing offers that were made earlier should intuitively lead to inefficient negotiation and, consequently, strategies involving such offers or withdrawing offers would not be subgame perfect equilibrium strategies.  This intuition is  confirmed by our technical results.
\end{itemize}

\item[The order of making offers.]  The order in which offers are made by the different players can be essential, especially in case of irrevocable offers. In  such cases we assume that the order in which players can make offers is set by a separate,  exogenous protocol which is an added component of the preplay negotiations game; for instance, it can be strictly alternating or random. Alternatively, the offers may be required to be made simultaneously by all players, as in  \cite{JW05} and  \cite{EP11} but we do not consider that option. 

\item[Rejection of offers.] Once made, offers may, or may not, be officially rejected before the play. A rejection by a player $B$ of an offer made to her by a player $A$ has the same practical effect as a withdrawal of the offer by $A$, but the choice to withdraw or not is now in the hands of $B$. Both options can be reasonable in different scenarios. 

\item[Conditionality of offers.] As discussed earlier, offers may be {\em unconditional}, i.e., not subject to acceptance or rejection by the player to whom the offer is made,  or {\em conditional} upon an expected (suggested or demanded) counter-offer by the player to whom the offer was made. Acceptance of a conditional offer means both acceptance of the offer and making the expected counter-offer. We emphasize that \emph{after acceptance, a conditional offer does not constitute a contract between the players turning the game into a cooperative one, but only a pair of unilateral offers, each binding only its proposer. It therefore transforms the current game into another non-cooperative game.} 
Rejection/withdrawal of a conditional offer means cancellation of both of the unconditional offers of which it consists.  
The option of rejection/withdrawal of conditional offers can be reasonably assumed under some circumstances (e.g. possibility for extended communication and for a low-cost negotiations), but not in others. We will consider both cases separately.
\end{description}

\section{Preplay offers and induced game transformations}
\label{sec:trans} 

In this section we describe the game transformations induced by preplay offers in a general and more technical fashion. 
\vcut{
(Subsection \ref{sec:general}). We discuss some of their basic properties and then formally 
extend the framework of preplay offers with conditional offers and briefly discuss withdrawals of offers \pcut{punishments} and offers contingent on strategy profiles in Subsection \ref{sec:offers}.
}

\subsection{Transformations of normal form games by preplay offers}\label{sec:general}
\begin{figure}[htb]\hspace*{\fill}%
\begin{game}{6}{4}
        & $B_1$ & $\cdots $ & $B_j$    & $\cdots $\\
       $A_1$ & $\cdots$ & $\cdots $ & $ a_{1j}, b_{1j}$    & $\cdots $\\
      $A_2$  & $\cdots$ & $\cdots $ & $a_{2j}, b_{2j}$    & $\cdots $\\
       $\cdots$ & $\cdots$ & $\cdots $ & $\cdots$    & $\cdots $\\
     $ A_i$  & $\cdots$ & $\cdots $ & $a_{ij}, b_{ij}$    & $\cdots $\\
       $\cdots$ & $\cdots$ & $\cdots $ & $\cdots$    & $\cdots $
  
\end{game}\hspace*{\fill}%
\caption{A general 2-player game}
\label{general}
\end{figure}

Here we formally define the notion of transformation induced by a preplay offers. For technical convenience we consider general 2-player game with a payoff matrix given in Figure \ref{general}; the case of N-player games is a straightforward generalization.
 
Suppose player $A$ makes a preplay offer to player $B$ to pay her additional utility\footnote{The reason we allow vacuous offers with $\alpha= 0$ is not only to have an identity transformation at hand, but also because such offers can be used by players as signaling, to enable coordination.
} $\alpha\geq 0$ if $B$ plays $B_{j}$. Recall that we denote such offer by $A \xlongrightarrow{\alpha \slash  B_{j}} B$. It transforms the payoff matrix of the game as indicated in Figure \ref{general-modified}.

\begin{figure}[htb]\hspace*{\fill}%
\begin{game}{5}{4}
        & $B_1$ & $\cdots $ & $B_j$    & $\cdots $\\
       $A_1$ & $\cdots$ & $\cdots $ & $ a_{1j}-\alpha, b_{1j}+\alpha$    & $\cdots $\\
      $A_2$  & $\cdots$ & $\cdots $ & $a_{2j}-\alpha, b_{2j}+\alpha$    & $\cdots $\\
       $\cdots$ & $\cdots$ & $\cdots $ & $\cdots$    & $\cdots $\\
     $ A_i$  & $\cdots$ & $\cdots $ & $a_{ij}-\alpha, b_{ij}+\alpha$    & $\cdots $\\
       $\cdots$ & $\cdots$ & $\cdots $ & $\cdots$    & $\cdots $
  
\end{game}\hspace*{\fill}%
\caption{A general 2-player game with an offer.}
\label{general-modified}
\end{figure}

We will call such transformation of a payoff matrix a {\bf primitive offer-induced transformation}, or a POI-transformation, for short. 

Several preplay offers can be made by each players. Clearly, the transformation of a payoff matrix induced by several preplay offers can be obtained by applying the POI-transformations corresponding to each of the offers consecutively, in any order. We will call such transformations {\bf offer-induced transformations}, or OI-transformations, for short. Thus, every OI-transformation corresponds to a \emph{set} of preplay offers, respectively a set of POI-transformations. 
Note that the set generating a given 
OI-transformation need not be unique, e.g., $A$ can make two independent offers 
$A \xlongrightarrow{\alpha_{1} \slash  B_{j}} B$
and 
$A \xlongrightarrow{\alpha_{2} \slash  B_{j}} B$
equivalent to the single offer 
$A \xlongrightarrow{\alpha_{1}+\alpha_{2} \slash  B_{j}} B$.

The general mathematical theory of OI-transformations is studied in more detail in \cite{VG-GT}. Here we only mention some observations about the game-theoretic effects of OI-transformations, which will be useful later on. 

\begin{enumerate}
\item An OI-transformation does not change the sum of the payoffs of all players in any outcome, only redistributes it. In particular, OI-transformations  preserve the class of zero-sum games. 

\item An OI-transformation induced by a preplay offer by player $A$ does not change the preferences of $A$ regarding her own strategies. In particular, (weak or strict) dominance between strategies of player  $A$ is invariant under OI-transformations induced by preplay offers of $A$, i.e.: a strategy $A_{i}$ dominates (weakly, resp. strongly) a strategy $A_{j}$ before a transformation induced by a preplay offer made by $A$ if and only if $A_{i}$ dominates (weakly, resp. strongly) $A_{j}$ after the transformation.  

\item The players can collude to make \emph{any} designated outcome, with any 
redistribution of its payoffs, a dominant strategy equilibrium, by exchanging sufficiently high offers to make the strategies generating that outcome with 
that redistribution of the payoffs, strictly dominant. 
\end{enumerate}

Thus, preplay offers can transform the game matrix radically.
However, we note that not every matrix transformation that preserves the sums of the payoffs in every outcome can be induced by preplay offers. In particular, this is the case if the transformed matrix differs from the original one in only one payoff. For general necessary and sufficient condition for a normal form game to be obtained from another by preplay offers see \cite{VG-GT}.

\medskip
A central question arising is what should be regarded as a \emph{solution} of a strategic game allowing binding preplay offers. The possible answers to that question crucially depend on the additional assumptions discussed earlier and on the procedure of 'preplay negotiations'; these will be discussed further.

\subsection{Extending preplay offers and OI-transformations}
\label{sec:offers}

\subsubsection{Conditional offers}
Unconditional offers always decrease the proponent's payoff at some outcomes, and hence making an unconditional offer comes with a cost. As in the Prisoners' Dilemma 2 Example \ref{example:PD2}, 
this can be a hindrance for making mutually beneficial offers and we will discuss this problem in more details in Section \ref{subsec:NoCond-NoW}. Furthermore, often in real life situations players who make such preplay offers expect some form of reciprocity from their fellow players and make their offers conditional on an expected `return of favour'.

For these reasons, we now extend the preplay offers framework to enable players to {\em suggest} a transformation of the starting game, by making a {\em conditional offer} to an opponent for payment subject to playing a certain strategy, in exchange for a similar `counter-offer' from that opponent. More precisely,  every conditional offer, denoted as $A \xlongrightarrow{\alpha \slash  \sigma_B \ \mid \  \beta \slash \rho_A} B$ is associated with a  {\bf suggested  transformation} of the starting game  $\mathcal{G}$ into a game $\mathcal{G}(X)$ where 
$X = \{A \xlongrightarrow{\alpha \slash  \sigma_B} B, B \xlongrightarrow{\beta \slash \rho_A} A\}$.

Two responses of the recipient of a conditional offer  $A \xlongrightarrow{\alpha \slash  \sigma_B \ \mid \  \beta \slash \rho_A} B$ are possible:  
it can be {\em accepted} or {\em rejected} by the player receiving it. If rejected, the offer is immediately cancelled and does not commit any of the players to any payment, and therefore it does not induce any transformation of the game matrix. If accepted, the {\bf actual transformation} induced by the offer is the suggested transformation defined above. 
Two important observations:
\begin{itemize}
\itemsep = -2pt
\item an unconditional offer has the same effect as an accepted conditional offer with a trivial counter-offer where $\beta=0$. 
\item a conditional offer can be seen as the proposal of two separate unconditional offers {\em that can only be enforced together}.
\end{itemize}

Conditional offers can be made to different players. Multiple conditional offers can be made to the same player, contingent upon same or different strategies of the recipient and the proposer, too. 

\subsubsection{Withdrawals of offers and transformations induced by them}

Withdrawal of an offer, i.e. a 'change of mind' by the player who makes the offer, can be simulated in a sense by matching the amount $\alpha$ offered by $A$ contingent on a given strategy $\sigma_B$ of $B$ by offers from $A$ to $B$ for the same amount $\alpha$, contingent on \emph{every other strategy of $B$}. However, his simulated offer withdrawal is costly for $A$ and, while the preferences on all outcomes remain the same for both players, the game is no longer the same. A {\em proper} withdrawal of $A$'s offer can only be achieved if $A$ extends his offer to cover all possible strategies of $B$ and $B$ offers in return to pay back the amount $\alpha$ to $A$ unconditionally, that is, makes offers of amount $\alpha$ to $A$, contingent on all strategies of $A$.

If a player $A$ withdraws an unconditional offer $A \xlongrightarrow{\alpha \slash  \sigma} B$ made earlier by her, the transformation $\mathcal{G}(X_A)$ of the game $\mathcal{G}$ induced by that offer must be reverted. The {\bf withdrawal} of a transformation $\mathcal{G}(X_A)$ is again a transformation, induced by the (fictitious) negative offer $A \xlongrightarrow{-\alpha \slash  \sigma} B$. 
Likewise, the transformation associated with withdrawal of an earlier made and accepted conditional offer $A \xlongrightarrow{\alpha \slash  \sigma \mid  \beta \slash \rho}  B$  consists of the reversal transformations of the two constituent unconditional offers. 

A withdrawal can thus be seen as a sort of unconditional reversal of payments that have already been enforced in a previous accepted offer. Thereby the possibility of performing a withdrawal strictly depends on the previous history of the negotiation and this feature will be of fundamental importance when treating preplay negotiations as full-fledged extensive games.

\section{Normal form games with preplay negotiations phase} 
\label{sec:normal-form}
\label{sec:preplay}

In  this section we first give some technical preliminaries and then introduce 
normal form games with preplay negotiations phase, first informally and then we formally define preplay negotiations games and discuss some features of them and the concept of efficiency of negotiation strategies.

\vcut{\vnote{To update} 
briefly discuss rationality assumptions and solution concepts for the normal form games. Then we model preplay negotiations as extensive games of perfect 
information and define generic notions of rationalilty and solution concepts for these preplay negotiations phase. In particular, we will define solutions of the preplay negotiations phase as outcomes of SPE strategy profiles. Finally, we discuss the combined solution concepts for the entire  games outlined above. 
}

\subsection{Preliminaries: solution concepts and values of normal form games} \label{sec:value}

We will be using $i,j, \ldots$ for variables ranging over players, while $A,B,\ldots$ will denote individual players. 

\subsubsection{Normal form games}
Let $\mathcal{G}=(N, \{\Sigma_i\}_{i\in N}, u)$ be a normal form game (hereafter abbreviated as NFG), where $N = \{1,\ldots,  n \}$ a finite set of players, $\{\Sigma_i\}_{i\in N}$ a family of strategies for each player and $u: N \times \prod_{i \in N} \Sigma_i \to \mathbb{R}$ is a {\bf payoff function} assigning to each player a utility for each strategy profile. The game is played by each player $i$ choosing a strategy from $\Sigma_i$. The resulting strategy profile $\sigma$ is the {\bf outcome} of the play and 
$u_{i}(\sigma) = u(i,\sigma)$ is the associated payoff for $i$. 
An outcome of a play of the game $\mathcal{G}$ is called {\bf maximal} 
 if it is a Pareto optimal outcome with the highest sum of the payoffs of all players.

\subsubsection{Solution concepts and solutions of normal form games}

Let $\bf{G}_N$ be the set of all normal form games for  a set of players $N$. By  {\bf solution concept for  $\bf{G}_N$} we mean  a map $\solc$ that associates with each $\mathcal{G} \in \bf{G}_N$ a non-empty set $\sol{\mathcal{G}}$ of outcomes of $\mathcal{G}$, called the {\bf\solc-solution of the game}. At times we will talk about players' strategies that are consistent with some solution concept. For a player $i$, we denote $\solc_i$ to be the restriction of the mapping $\solc$ to $i$ returning, instead of full outcomes, only strategies of player $i$ consistent with $\solc$ in the sense that $\solc_i({\mathcal{G}})= \{\sigma_i \in \Sigma_i \mid \sigma \in \sol{\mathcal{G}}\}$. Slightly abusing notation we will also consider mappings of the form $\solc_{-i}$ to indicate the  mapping $\sol{\mathcal{G}}$ restricted to player $i$'s opponents. Solution concepts formalize the concepts of rationality of the players in the strategic games. A \solc-solution of a strategic game $\mathcal{G}$ basically tells us what outcomes of the game the players could, or should, select in an actual play of that game, if they adopt the solution concept \solc. 

In this work we do not commit to a specific solution concept for the normal form games but we assume that the one adopted by the players satisfies the necessary condition that \emph{every outcome in any solution prescribed by that solution concept must survive iterated elimination of strictly dominated strategies}. We will call such solution concepts {\bf acceptable}. This condition  reflects the assumption that players would never play strategies that are dominated, and that this exclusion is a common knowledge amongst them and can be used in their strategic reasoning. Thus, the weakest acceptable solution concept is the one that returns all outcomes surviving iterated elimination of strictly dominated strategies.

Games for which the solution concept \solc returns a single outcome will be called {\bf\solc-solved}. For instance, every game with a strongly dominating strategy profile is 
\solc-solved for any acceptable solution concept  \solc. 
Games for which  \solc  returns only maximal outcomes will be called {\bf optimally \solc-solvable}. If for every player all these maximal outcomes provide the same payoffs, we call the game   {\bf perfectly  \solc-solvable}. Games that are \solc-solved and perfectly \solc-solvable (i.e., \solc returns one maximal outcome) will be called {\bf\solc-perfectly solved}. 

The ultimate objective of a preplay negotiation is to transform the starting NFG into a perfectly \solc-solvable one. Ideally, it should be a \solc-perfectly solved one, but this is not always possible: cf. any symmetric Coordination game. 

\subsubsection{Players' expected values of a game}  
It is necessary for the preplay negotiation phase that will be introduced later for each player to have an {\bf expected value} of any NFG
that can be played. Naturally, that expected value would depend not only on the game but also on the adopted solution concept 
and on the player's level of risk tolerance.  A risk-averse player would assign as expected value the minimum of his payoffs over all outcomes in the respective solution, while a risk-neutral player could take the probabilistic expected value of these payoffs, etc. Note that the expected value of any \solc-solved game for any player $i$ naturally should equal the payoff for $i$ from the only outcome in the solution. 

For sake of definiteness, unless otherwise specified further, we adopt here the conservative, risk-averse  approach  and will define for every acceptable solution concept \solc, game  $\mathcal{G}$ and a player $i$, the expected value of $\mathcal{G}$ for $i$ relative to the solution concept \solc to be: 
\[ \val^{\solc}_{i}(\mathcal{G}) = \max_{\sigma_i \in \solc_i({\mathcal{G}})} \min_{\sigma_{-i} \in \solc_{-i}({\mathcal{G}})}  u_{i}(\sigma) 
\]

\subsection{Normal form games with preplay negotiations phase informally} 

Our setting for normal form games with preplay offers begins with a given `starting' normal form game $\mathcal{G}$ and  consists of two phases:

\begin{itemize}
\itemsep = -2pt
\item A {\em preplay negotiation phase}, where players negotiate on how to transform the  game $\mathcal{G}$ by making unconditional offers, accepting or rejecting conditional offers they receive, and possibly withdrawing old ones. This phase constitutes an extensive form game, which we call a \emph{preplay negotiation game} (PNG). 

\item An {\em actual play}  phase where, after having agreed on some OI-transformation $X$ in the previous phase, the players play the resulting game $\mathcal{G}(X)$. 
\end{itemize}

Players engage in pre-play negotiations with the purpose of reaching a best {for them} possible agreement based on OI transformation of the original game $\mathcal{G}$. Major questions that we set out to study are: 

\begin{itemize}
\itemsep = -2pt
\item \textit{What constitutes an optimal/rational/efficient negotiation strategy and what are the expected outcome(s) when players follow such strategies? }
\item  \textit{In particular, when can players agree upon Pareto optimal outcomes in their preplay negotiations if playing rationally?}
\item  \textit{What can, or should, players agree upon in the preplay negotiations phase when the original game has several Pareto optimal outcomes?}
\end{itemize}

Further we introduce, first informally and then fully formally, the setup of PNGs as extensive-form bargaining games, including the concepts of moves and histories, the order of moves, the possibility of players come to a disagreement, and finally a notion of solution for these games.

\subsection{Moves, histories and preplay negotiations games}\label{sec:moves}

Depending on some of the optional assumptions, the players can have several possible moves in the preplay negotiations phase. Let us consider the most general case, where both conditional offers and withdrawals of offers are allowed. Then the moves available to the player whose turn is to play depend on whether or not he has received any conditional offers since his previous move. If so, we say that the player has {\bf pending conditional offers}. The possible moves of the player in turn are as follows. 

\begin{enumerate}
\item If the player has no pending conditional offers, he can: 
\begin{enumerate}
\item \emph{Make an offer} (conditional or not). 
\item \emph{Pass}.
\item (Optional)
\emph{Withdraw an offer} he has made at a previous move. 
\item (Optional) \emph{Opt out} (see Section \ref{subsec:disagree}).
\end{enumerate}
\item If the player has pending conditional offers, for each of them he can: 
\begin{enumerate}
\item \emph{Accept the pending offer} by making the requested counter-offer to the player who has made the conditional offer, and then make an offer of his/her own or pass or opt out (when available).
\item \emph{Reject the pending offer}, and then make an offer of his/her own or pass or opt out (when available).
\end{enumerate}
\end{enumerate}

If \emph{all players have passed} at their last move, or \emph{any player has opted out}, the preplay negotiations game is over. 

\vcut{
Note that while conditional offers always require a response (acceptance or rejection) by the player receiving them, this is not the case for unconditional offers. The effect of the latter ones is to immediately update the normal form game as specified by the offer.
}

We say that an offer of the game is {\bf passing} if its acceptance by the opponents is followed by a pass of the proponent. In other words, the one making the offer would be happy to end the game with the suggested transformation. Likewise, an acceptance is passing if, once declared, it is followed by a pass move of the same player. In other words, with a passing acceptance a player declares agreement to terminate the game with the proposed transformation. 
When opting out is not allowed, 
passing moves (i.e. offers or acceptances that are passing), are the only way for players to terminate the game in agreement and the only way to effectively deviate from undesired outcomes.

We now define the notion of a {\bf history} in the preplay negotiations phase as a finite or infinite sequence of admissible moves by the players who take their turns according to an externally set protocol (see further). 
Every finite history in such a game is associated with the {\bf current NFG}: the result of the OI-transformation of the starting game by all offers that are  so far made, accepted (if conditional) and currently not withdrawn. 
The current NFG of the empty history is the input NFG of the preplay negotiations game.

 A {\bf play} of a preplay negotiations game is any finite history at the end of which the preplay negotiations game is over, or any infinite history. 

In order to eventually define realistic solution concepts for preplay negotiations games we need to endow every history in such games with value for every player. Intuitively, {\bf the value of a history} is the value for the player of the current NFG associated with that history in the case of non-valuable time, and the same value accordingly discounted in the case of valuable time. 

Now, a {\bf preplay negotiation game (PNG)} can be defined generically as a turn-based, possibly infinite, extensive form game that starts with an input  NFG 
$\mathcal{G}$ and either ends with a transformed game 
$\mathcal{G}'$ or goes on forever, which we discuss further. The 
{\bf outcome of a play of the PNG} is the resulting  transformed game 
$\mathcal{G}'$ in the former case and 'Disagreement' (briefly $D$) in the latter case. 


\subsection{Preplay negotiation games formally} 


Here we provide a formal definition for the general N-player case of 
preplay negotiation games.

\begin{definition}[Preplay negotiation game] 
\label{def:PNG}\label{def:constraints}

A {\em preplay negotiation game}  is a tuple 
$\mathcal{E} = (\N, \mathcal{G}, \solc, 
\act, \hist, \turn,\{\Sigma_i\}_{i\in\N}, \bfg, \out, \bfu)$, where:

\begin{itemize}
\itemsep = -2pt
\item N is the set of players.

\item $\mathcal{G}$ is the starting normal form game. 

\item $\solc$ is an acceptable solution concept for normal form games. 

\item  $\act$ is a set of {\bf actions}, or {\bf moves} of types as discussed earlier. 
\item $\hist$ is a non-empty set of finite or infinite sequences of actions, called {\bf histories}, that  includes the empty sequence $\epsilon$ and is \emph{prefix-closed}, meaning that every prefix of a history in $\hist$ belongs to $\hist$, and \emph{limit-closed}, meaning that the infinite union of a chain by extension of finite histories  in $\hist$ belongs to $\hist$, too.  

A history $h \in \hist$ is {\b fterminal in $\hist$} if it is infinite or there is no history in $\hist$ extending it. 
The set of terminal histories in $\hist$ is denoted by $\trmhist$ and  the set of finite histories in $\hist$ by $\fhist$. 

For $h,h' \in \fhist$ and $o \in \act$ we denote by $h;o$ the extension of $h$ with the action $o$ and by $h;h'$ the concatenation of $h$ with $h'$.  XX

\item $\turn: \hist \setminus \trmhist \to N$ is the {\bf turn function}, assigning the players who are to move at non-terminal histories. We denote $\hist _{i}:= \turn^{-1}(i)$ for each $i\in N$ the set of histories where it is $i$'s turn to play. 

Here we assume that the turned function is exogenously defined, e.g. in some fixed cyclic order or depending on the last move made.  

\item 
$\Sigma_i$, for each $i\in N$, is a non-empty set of  {\bf strategies}  
$\sigma_i : \hist _{i} \to \act$ that assigns an action for $i$ to any non-terminal history in $\hist _{i}$. 


\item $\bfg: \hist \to {\bf G_N}$ is a function associating to each finite history the {\bf currently accepted NFG}, defined  below. 
 
\item $\out: \prod_{i \in N} \Sigma_i \to \trmhist$ is an
 {\bf outcome play function},  assigning to each strategy profile $\sigma$ the terminal history $\out({\sigma})$ generated by $\sigma$.  
 
 Respectively, the  {\bf outcome NFG} of $\sigma$ is 
 $\bfg(\out({\sigma}))$. 
 
\item 
$\bfu: \N \to (\trmhist \to \mathbb{R})$ is the {\bf utility function of the PNG}, associating to each player the payoff function $\bfu_i$ such that $\bfu_i(z)={\bf v^{\solc}}_i(\bfg(z))$ for every  finite $z\in \trmhist$. 
Further, for $z,z^{\prime} \in \trmhist$, with $z$ finite and $z^{\prime}$ infinite, we require that  $\bfu_i(z) \geq \bfu_i(z^{\prime})$ for all players $i$, and $\bfu_j(z) > \bfu_j(z^{\prime})$ for some $j$, i.e., no disagreement is better for all players than any agreement.
\end{itemize}

Now we define the function $\bfg$. Its intended meaning is that $\bfg(h)$ would be the outcome of the PNG if the game ended at $h$. Its precise definition depends on the repertoire of moves that are allowed in the PNG, as follows: 


\begin{itemize}
\itemsep = -2pt
\item $\bfg(\epsilon)$ is the starting normal form game $\mathcal{G}$. 
\item If $h =  h';o$, where the last move $o$ is an unconditional offer, then $\bfg(h) = \bfg(h')(o)$, i.e., the transformation of $\bfg(h')$ by the offer $o$. 

\item If $h =  h';a$, where the last move $a$ is an acceptance of a conditional offer $o$, then $\bfg(h) = \bfg(h')(o)$. 

\item If $h =  h';o;h'';w$, where $w$ is a withdrawal of the offer $o$ then $\bfg(h)=\bfg(h';h'')$, i.e. the transformation by the withdrawn offer is reverted. 

\item In all other cases of actions $a$, $\bfg(h;a) = \bfg(h)$. 
\end{itemize}

\end{definition}

\textbf{Solution of PNG.}  
By solution of a PNG we mean \emph{the set of all transformed normal form games $\bfg(h)$ for all outcomes $h$ of plays effected by subgame perfect equilibrium (SPE) strategy profiles in the PNG}.

\vcut{
\paragraph{On the order of making moves} 
Depending on some of the optional assumptions, the players can have several possible moves in the preplay negotiations game, which they can make simultaneously, in several rounds, or by taking turns according to some externally set protocol or by a randomized procedure. We will focus on turn-based negotiations games but 
We will not discuss here how the order of making moves is determined. We only state that when time is not valuable the order in which the players take turns to make their offers is irrelevant for the eventual solution, \vnote{I disagree here. See examples in the case of unconditional offers.} and multi-round negotiations with simultaneous offers are reducible to multi-round turn-based negotiations with any order of making moves. \vnote{ I do not see this, either. To discuss.} 
}

\subsection{Disagreements}
 \label{subsec:disagree}
 
Clearly, players would only be interested in making preplay offers inducing payoffs that are ``optimal'' for them. Therefore, rational players are expected to ``negotiate'' in the preplay phase the play of Pareto optimal outcomes. In particular, if the game has a unique strictly Pareto dominant outcome then the players can negotiate a  transformation of the game to make it the (unique) dominant strategy equilibrium. Yet, players that are getting lesser shares of the total payoff may still want to negotiate a redistribution, so even in this case the outcome of the preplay negotiations is not a priori obvious. In particular, there is no guarantee that the PNG will ever terminate, i.e. that its solution is non-empty.  

The PNG may terminate if all players pass at some stage,  in which case we say that the players have reached agreement, or may go on forever, in which case the players have failed to reach agreement; we call such situation a \emph{(passive) disagreement} and we denote any such infinite history with $D$. We will not discuss disagreements and their consequences here, but will make the explicit assumption that \emph{any agreement is better for every player than disagreement} in terms of the payoffs, by assigning payoffs of $-\infty$
in the entire game for each player if the PNG evolves as a disagreement.
However, we also outline a more flexible and possibly more realistic alternative, whereby players can explicitly express tentative agreements with the status quo before every move they make, essentially by saying ``\textit{So far so good, but let me try to improve the game further by offering \ldots}'', or express disagreements, by  
essentially saying ``\textit{No, I am not happy with the way the negotiations have 
developed since the last time I agreed, so I'd like to improve the game 
by offering instead \ldots}''. This type of negotiations involves, besides the other moves listed above, also formal statements of acceptance or non-acceptance of the current NFG, where the input NFG is automatically accepted by all players and at every stage of the negotiations, the  current NFG is the one on which they are currently negotiating by making offers, whereas 
the currently accepted NFG is the last current one for which all players have explicitly stated acceptance. Then if at any stage of the PNG any player is currently unhappy and realizes that he cannot improve further because of the other players not willing to accept his best conditional offers, then he can terminate the negotiations by explicitly {\bf opting out}, 
which would leave as an outcome game the currently accepted NFG.


\subsection{Preplay negotiations games and assumptions on players' rationality} 



In order to understand how solutions of preplay negotiation games  look we need to understand the 
equilibria of PNGs. This seems a very complex problem and its analysis crucially depends on the specific optional assumptions that we make regarding the types of allowed moves, value of time, and most importantly -- the players' common rationality assumptions in the PNG. 

While our analysis of games with conditional offers will be based on a standard SPE analysis --- which allows for a direct connection with bargaining games --- for the case of unconditional offers we discuss and adopt what we call an
\emph{immediate rationality assumption}, i.e., we study players that calculate optimal offers on a given NFG without considering the possible extensions of the preplay negotiation game after the immediate expected response to their move. The difference is that in the general case a player may afford making a `sub-optimal' move in the play of a PNG, transforming the currently accepted game into one with a lesser value for that player, with the expectation, justified by a long-term rationality assumption, that the opponent will not opt out but will continue the negotiation for the sake of reaching a mutually better outcome, reasoning likewise.  
On the other hand, in the restricted case of \textit{Immediate Rationality Assumption}, hereafter abbreviated as IRA,
 an optimal strategy of a player would only prescribe moves that would guarantee that the resulting transformed NFG has a no lesser value for the player making that move than the currently accepted NFG. To put it simple, IRA implies that players are short-sighted and prescribes to them to play optimal, but `locally safe'  strategies in the PNG. This assumption is often justified, e.g., when players have no a priori knowledge about each rationality and patience and also makes the analysis somewhat easier, but by no means trivial, as we will see further. 
As the analysis in the case of many-player PNGs is still very complicated and cannot be presented in a single paper, hereafter we restrict attention to the 2-player case. In order to carry our such analysis and to make statements about existence of `good' solutions assuming IRA, we first need to discuss the notions of `feasibility of moves' and `efficiency of negotiation strategies'. In this context we will use the term  ``efficient''  not in its standard game-theoretic sense, i.e., by applying it to outcomes, but to the way outcomes are reached.

\subsubsection{Feasible offers and moves} 

In principle, players can make offers that would induce transformations decreasing their expected value of the game. Generally, such offers would not be rational to make under the IRA, but they may be still be admissible in some circumstances, e.g., when conditional offers are not allowed but withdrawals of offers are. We say that a player's offer is {\bf weakly feasible} if it does not decrease that player's expected value of the game in the game transformed by that offer; the offer is {\bf feasible} if it strictly increases that expected value. 
This is a generic notion of feasibility of offers, which needs to be extended further to the notion of feasible moves in the PNG. The latter is specific to some of the optional additional assumptions which we will discuss in more detail for the 2-player case in the next section. 
We argue that, assuming IRA, in order for a player's strategy in the preplay negotiation phase to be a part of a rational solution, 
it  must only involve weakly feasible moves. 
%

\subsubsection{Minimal offers}

With their preplay offers players want to create incentives for the other players to play desired strategies. So, feasibility is a necessary condition for 
an offer to be made in an actually played PGM, but it is not sufficient for it to be a part of a subgame perfect equilibrium strategy. Clearly, an optimal offer from a player to another would be a \emph{minimal feasible one} providing a sufficient incentive for the recipient of the offer to play the desired transformation, but not more than that. The question of what is a minimal offer that achieves such objective crucially depends on the adopted solution concept and, in particular, on the rationality assumptions and reasoning skills of the recipient. For instance, if the players know the solution of the starting normal form game $\mathcal{G}$,  
induced by the adopted solution concept, then they also know which outcomes can be selected among the ones surviving the iterated elimination process. Thereafter, if a player $A$ wants to induce with a preplay offer another player $B$ to play a given strategy $\sigma_B$ then, for any acceptable solution concept, it would suffice for $A$ to make any sufficiently large offer that would turn $\sigma$ into a strictly dominant strategy for $B$. But, such offer may be prohibitively costly or, depending on the solution concept and the rationality assumptions for $B$, unnecessarily generous. For instance, when a player $B$ receives an offer 
$A \xlongrightarrow{\delta \slash  \sigma_B}B$, he should naturally expect that $A$ considers playing $A$'s best response to  $\sigma_B$, so $B$ can anticipate the outcome of the transformed game, and if $B$ considers that outcome better than his current expected value, that should suffice for $A$'s offer to work.

A technical detail: 
it is often the case that no minimal offer exists that guarantees to achieve the objective, e.g., to turn the desired strategy  into a \emph{strictly}  dominant one. For instance, if it suffices for $A$ to pay to $B$ any amount that is \emph{greater than $d$} for that purpose, then any offer of $d + \epsilon$, for $\epsilon > 0$, should do. Clearly, however, there is a practical minimum beyond which a player in question would not bother optimizing any further, so we will often refer to offers of payments $d^{+}$ meaning $d + \epsilon$ for `sufficiently small $\epsilon > 0$'  without specifying the value of $\epsilon$, but still allowing its further reduction, as long as it remains strictly positive.

\subsubsection{Efficient negotiation strategies} 

\begin{definition}[Efficient negotiation strategies]
A strategy in the PNG 
is an {\em efficient negotiation strategy} if it only involves making (minimal) feasible offers, it passes once they are accepted, and  -- in the case when conditional offers are allowed -- 
at no point prescribes withdrawal of earlier made offers.  It is {\em strongly efficient} if the vector of payoffs of the outcome it attains is a redistribution of the vector of payoffs of a maximal outcome.
\end{definition}

A number of important relevant questions arise: 
\begin{itemize}
\itemsep = -2pt
\item Is it the case that every subgame perfect equilibrium (SPE) strategy of a PNG is an efficient negotiation strategy and vice versa?
\item If not, can the inefficient ones be replaced by efficient ones generating the same, or at least as good solution? 
\vcut{
\item Is it the case that, when generalized offers, consisting of sets of unconditional offers, are allowed there are subgame perfect equilibria where each player need not make more than one offer (combining all intended subsequent offers by that player)?
Does this depend on the number of players? 
}
\item Under what conditions can a given (maximal) Pareto optimal outcome in the starting NFG become the unique outcome of the final NFG? 
\end{itemize}

To answer these questions we need an analysis of the solutions of the PNG game.  Further we provide such partial analysis for the case of two players. 

\section{Two-player preplay negotiation games with unconditional offers}
\label{sec:2-player-uncond}

\label{subsec:NoCond}


We begin with the case of more restricted preplay negotiations, where conditional offers are not possible, or not allowed. As we will see further, the strategic reasoning in such preplay negotiations games is rather different from the case with conditional offers, because any player who makes an unconditional offer puts himself in a disadvantaged position by offering unilaterally a  payment to the other player and thus transforming the payoff matrix to the other player's advantage. Therefore, generally, players are more interested in receiving, rather than in making, unconditional offers and this affects essentially the preplay negotiations phase.  

According to the Immediate Rationality Assumption IRA, here we focus on the \emph{locally rational} behavior of players exchanging unconditional offers, by first determining the \emph{best} (for the offerer) rational unconditional offer that a player can make on a given 2-player NFG. Then we illustrate with some examples possible evolutions and outcomes of the preplay negotiation phase consisting of exchanging such best offers and draw some conclusions. In other words, here we analyze and illustrate the rationality of {\em moves}, rather than full-blown strategies, suggesting that every {good} notion of IRA-compliant equilibrium used to analyze PNGs without conditional offers should take this rationality into account. 
We leave untreated for now the question of how the value of time affects the outcomes of the preplay negotiations games in this case, by tacitly assuming that time is not valuable.

\subsection{The effect of allowing withdrawals of unconditional offers}
\label{subsec:NoCond-W}

We first argue that when withdrawals of unconditional offers are allowed, conditional offers can be simulated, too, even though at the cost of some time delay. Indeed, if player $A$ wants to make a conditional offer $A \xlongrightarrow{\alpha \slash  B_{j} \ \mid \ \beta \slash  A_{i}} B$ she can make the unconditional offer $A \xlongrightarrow{\alpha \slash  B_{j}} B$ expecting the matching (or better) unconditional offer $B \xlongrightarrow{\beta \slash  A_{i}} A$ from $B$. How can the receiver $B$ guess the expected matching offer, if side communication is not possible or not allowed? Note that the offer $A \xlongrightarrow{\alpha \slash  B_{j}} B$ has 2 effects: it changes the payoff table in a way beneficial for $B$ and indicates that player $A$ wants player $B$ to play $B_{j}$. Therefore, $B$ can naturally expect that (disregarding for a moment all other offers) $A$ intends to play her best response to $B_{j}$.  However, an offer from $B$ to $A$ may change $A$'s best response to $B_{j}$ in a way, that would make it more beneficial for $A$, and at least as beneficial for $B$, if $A$ plays another strategy, say $A_{i}$. By inspecting the possibilities $B$ can identify his options for matching offers that would make $A$'s unconditional offer worth her while. If $B$ has more than one such options, he can guess and try. 
If the expected matching offer is not received in the next round of the preplay negotiations, $A$ can subsequently withdraw her offer, thus indicating that her expectations were not met, but later can make it again, possibly repeating this `ritual' until $B$ eventually realizes what is expected from him and offers it (or until $A$ gives up expecting). 
Thus, the case of unconditional offers with withdrawals is essentially reducible to the case where conditional offers are allowed, treated further. We only note here that when time is valuable the simulation suggested above may be costly and leading to side effects. 

\subsection{Preplay negotiations with unconditional offers and no withdrawals}
\label{subsec:NoCond-NoW}

The case when no withdrawals of offers are allowed is essentially different. As we will see further, in this case the players can be genuinely disadvantaged by making the first offer, and this can be crucial for the outcome of the negotiations. We can distinguish 3 types of unconditional offers: 

\begin{enumerate}
\item {\bf vacuous offers}, of the kind $A \xlongrightarrow{0 \slash \sigma} B$ for payment of 0. These can be used instead of  passing, but also, more importantly, as a kind of \emph{signaling}, i.e., indication that $A$ expects $B$ to play $\sigma$, for breaking the symmetry in case of symmetric games with several equivalent optimal equilibria. 
\item {\bf $\epsilon$-offers}, of the kind $A \xlongrightarrow{\epsilon \slash \sigma} B$ for a small enough $\epsilon > 0$. These can be used similarly, for breaking the symmetry, when $B$ has more than one best for him moves which,  however, yield different payoffs for $A$. Using such a move, $A$ can make any of these strictly preferable for $B$ and, thus, can turn a weak equilibrium into a strict one, with minimal cost. 
\item {\bf effective offers}, of the kind $A \xlongrightarrow{d \slash \sigma} B$ for a (large enough) $d>0$. These are the standard offers used to change the recipient's preferences and influence his choice of strategy in the 2nd phase. 
\end{enumerate}

It is easy to see that in ideal two-players negotiations none of them needs to make two consecutive offers, between which the opponent has passed or made a vacuous offer. Indeed, no player would be better off by making offers  in the same game contingent on two or more different strategies of the opponent; in fact, such multiple offers send to the opponent confusing signals. Furthermore, two or more offers by the same player that are contingent on the same strategy of the opponent can be combined into one.  
So, leaving aside the question of who starts the preplay negotiations game, in the case where only unconditional and irrevocable offers are allowed, the PNG  
consists of a sequence of alternating offers made in turn by the two players until both of them pass. 
Thus, in order to capture the notion of efficient negotiations in this case, we need to analyze the question of \emph{what are the best unconditional and irrevocable offers that a player can make on a given NFG}?  

\subsection{Computing the best unconditional offers of a player} 

What is an IRA-based rational player's reasoning when considering making an unconditional and irrevocable offer to another player in a given NFG $\mathcal{G}$? Suppose, player $A$ considers making such an offer to player $B$. Then, for each strategy $B_{j}$  of  $B$, player $A$ considers making an offer contingent on $B$ playing  $B_{j}$. To make sure that $B$ will play $B_{j}$ in the resulting game, it suffices to make the latter a strictly dominant strategy for $B$. The necessary payment for that, however, can be prohibitively high for $A$ because after that payment $A$'s best response to $B_{j}$ may yield a worse payoff than the current (e.g., maxmin) expected value  for $A$ of the original game. So, a more subtle reasoning is needed, presented by the following procedure. 

\begin{enumerate}
\item For each strategy $B_{j}$  of  $B$, player $A$ looks at her best response to  $B_{j}$. Suppose for now that it is unique, say  $A_{i_{j}}$. Then, this is what   $B$ would expect $A$ to play if $B$ knows that $A$ expects $B$ to play $B_{j}$. In this case, $A$  computes the minimal payment needed to make $B_{j}$ not necessarily a strictly dominant strategy, but a \emph{best response to $A_{i_{j}}$}, i.e., the minimal payment that would make the strategy profile  
$\sigma_{i_{j},j} = (A_{i_{j}},B_{j})$ a Nash equilibrium. That payment is 
\[\delta^{A}_{i_{j},j} = \max_{k} (u_{B}(A_{i_{j}},B_{k}) - u_{B}(\sigma_{i_{j},j})).\]
If it is positive, or is 0 but reached not only for $k=j$ but also for other values of $k$, then, in order to break $B$'s indifference and make $\sigma_{i_{j},j}$ a \emph{strict} Nash equilibrium, $A$ has to add to $\delta^{A}_{i_{j},j}$ a small enough $\epsilon > 0$, thus eventually producing  the minimal necessary payment $\delta^{A}_{j}$.

\item If $A$'s best response to $B_{j}$ is not unique, then $A$ should compute the minimal payment $\delta^{A}_{j}$ needed to make $B_{j}$ the best response of $B$ to \emph{each}  of $A$'s best responses to $B_{j}$. Clearly, that should be the maximum of all $\delta^{A}_{i_{j},j}$ computed above, possibly plus  a small enough $\epsilon > 0$. 

\item Once $\delta^{A}_{j}$ is computed, $A$ computes her expected payoff in the transformed game $\widehat{\mathcal{G}}_{B_{j}}$ after an offer $A \xlongrightarrow{\delta^{A}_{j} \slash B_{j}} B$, which is: 
\[v^{A}(\widehat{\mathcal{G}}_{B_{j}})= u_{A}(\sigma_{i_{j},j}) - \delta^{A}_{j}. \]

\item Finally, $A$ maximizes over $j$: 
\[v^{A}(\widehat{\mathcal{G}}) = \max_{j} v^{A}(\widehat{\mathcal{G}}_{B_{j}}).\]

If the maximum is achieved for more than one $j$, then $A$ can choose any of them, or --better -- the one yielding the least payoff for $B$, thus stimulating $B$ to make her a further offer. 

If this maximum is 0 and reached for only one value of $j$, then there is no need for $A$ to make any offer, because in this case there is a unique Nash equilibrium in the game and $A$ cannot make any offer that would improve on her payoff yielded by that Nash equilibrium. If the maximum is 0, but reached for more than one values of $j$, then $A$ must still make a vacuous offer $A \xlongrightarrow{0 \slash B_{j}} B$ in order to indicate to $B$ for which Nash equilibrium she will play. 
\end{enumerate}

The reasoning for $B$ is symmetric, eventually producing the value $v^{B}(\widehat{\mathcal{G}})$.

The definition of $v^{A}(\widehat{\mathcal{G}})$ implies the following:

\begin{proposition}
Given the NFG $\mathcal{G}$, the value 
$v^{A}(\widehat{\mathcal{G}})$ is the best payoff that player $A$ can guarantee as a result of the players playing any Nash equilibrium induced by an unconditional offer from $A$ to $B$ in the transformed game.
\end{proposition}

It is now up to player $A$ to decide whether to make the respective offer leading to the value $v^{A}(\widehat{\mathcal{G}})$ -- if that offer would improve her current expected value -- or to pass, possibly by making only a vacuous offer, for the sake of indicating to $B$ on which of the several equivalent Nash equilibria to coordinate (as in the symmetric coordination game), when appropriate. 

\begin{example}[Solving a game by exchange of unconditional offers]
~ 

Consider the following NFG $\mathcal{G}$ between players $R$ (row) and $C$ (column):  
\[
\begin{game}{2}{3}
            & $C1$  & $C2$ & $C3$ \\
$R1$   & $2,10$   &$10,4$ & $5,1$\\
$R2$   & $6,0$   &$4,4$ & $6,3$ \\
\end{game}\hspace*{\fill}%
\]

This game has no  pure strategy NE. The maxmin solution is $(R2,C2)$ with payoffs $(4,4)$, which is not Pareto optimal.

Suppose, player $R$ is to make the first offer. Let us compute the best offer that $R$ can make to $C$.  (We will often denote $d+\epsilon$ by $d^{+}$ and $d-\epsilon$ by $d^{-}$.) 

\begin{itemize}
\itemsep = -2pt
\item The best response of $R$ to $C1$ is $R2$. \\
So, $\delta^{R}_{2,1} = 4 - 0 +  \epsilon = 4^{+}$ 
and 
$v^{R}(\widehat{\mathcal{G}}_{C1})= 6 - 4^{+}= 2^{-}$ 

\item  The best response of $R$ to $C2$ is $R1$. \\
So, $\delta^{R}_{1,2} = 10 - 4 +  \epsilon = 6^{+}$ and 
$v^{R}(\widehat{\mathcal{G}}_{C2})= 10 - 6^{+} = 4^{-}$. 

\item  The best response of $R$ to $C3$ is $R2$. \\
So, $\delta^{R}_{2,3} = 4 - 3 +  \epsilon = 1^{+}$ and 
$v^{R}(\widehat{\mathcal{G}}_{C3})= 6 - 1^{+} = 5^{-}$. 
\end{itemize}

Thus, $v^{R}(\widehat{\mathcal{G}})= v^{R}(\widehat{\mathcal{G}}_{C3}) = 5^{-}$, meaning that $R$'s best offer to $C$ is $R \xlongrightarrow{ 1^{+} \ \slash \ C3} C$.
The resulting transformed game is 
\[
\begin{game}{2}{3}
            & $C1$  & $C2$ & $C3$ \\
$R1$   & $2,10$   &$10,4$ & $4^{-},2^{+}$\\
$R2$   & $6,0$   &$4,4$ & $5^{-},4^{+}$ \\
\end{game}\hspace*{\fill}%
\]
It has one Nash equilibrium $(R2,C3)$ yielding payoffs $(5^{-},4^{+})$ 
which are strictly better than the players maxmin values, but not yet Pareto optimal. 

Now, let us compute the best offer of $C$ to $R$ in the transformed game.

\begin{itemize}
\itemsep = -2pt
\item The best response of $C$ to $R1$ is $C1$ and 
$\delta^{C}_{1,1} 
= 4^{+}$. 
So, $v^{C}(\widehat{\mathcal{G}}_{R1})= 10 - 4^{+} = 6^{-}$.  

\item  The best response of $C$ to $R2$ is $C3$ and $\delta^{C}_{2,3} = 0$.
Thus, $v^{C}(\widehat{\mathcal{G}}_{R2})= 5^{-}$.  
\end{itemize} 

So, $v^{C}(\widehat{\mathcal{G}})=  6^{-}$, which is better than $C$'s current value of $4^{+}$. Thus, $C$ can improve his value by making the offer $C \xlongrightarrow{4^{+} \ \slash \ R1} R$. The resulting transformed game is 
\[
\begin{game}{2}{3}
            & $C1$  & $C2$ & $C3$ \\
$R1$   & $6^{+},6^{-}$   &$14^{+},0^{-}$ & $8,-2$\\
$R2$   & $6,0$   &$4,4$ & $5^{-},4^{+}$ \\
\end{game}\hspace*{\fill}%
\]
It has one Nash equilibrium $(R1,C1)$, where the strategy $R1$ is strictly dominant for $R$, yielding payoffs $(6^{+},6^{-})$ which are strictly better than the previous ones of  $(5^{-},4^{+})$, but not yet Pareto optimal.  So, let us see whether $R$ can improve any further the resulting game, given that the strategy $R1$ is already his best response to all strategies of $C$:    

\begin{itemize}
\itemsep = -2pt
\item For $C1$:  
$\delta^{R}_{1,1} = 0$ and $v^{R}(\widehat{\mathcal{G}}_{C1})= 6^{-}$  

\item For $C2$:  $\delta^{R}_{1,2} = 6^{-} - 0^{-} +  \epsilon = 6^{+}$ and 
$v^{R}(\widehat{\mathcal{G}}_{C2})= 14^{+} - 6^{+}  = 8$. 

\item For $C3$: 
$\delta^{R}_{1,3} = 6^{-} - 2 +  \epsilon = 8$ and 
$v^{R}(\widehat{\mathcal{G}}_{C3})= 8 - 8 = 0$. 
\end{itemize}

Thus, $v^{R}(\widehat{\mathcal{G}})= v^{R}(\widehat{\mathcal{G}}_{C2}) = 8$,
which is better than $R$'s current value of $6^{+}$, hence $R$'s best offer to $C$ now is $R \xlongrightarrow{6^{+} \ \slash \ C2} C$.
The resulting transformed game is 
\[
\begin{game}{2}{3}  
            & $C1$  & $C2$ & $C3$ \\
$R1$   & $6^{+},6^{-}$   & $8,6$  & $8,-2$\\
$R2$   & $6,0$               & $-2^{-},10^{+}$ & $5^{-},4^{+}$ \\
\end{game}\hspace*{\fill}%
\]
It has a strictly dominant strategies equilibrium $(R1,C2)$ yielding payoffs $(8,6)$ which are strictly better than the previous ones $(6^{+},6^{-})$. In fact, this is the only Pareto maximal outcome in the game, and one can now check that none of the players can make any further improving offers. Thus, this is the end of the negotiation phase. 

We leave to the reader to check that, if $C$ makes the first offer, the negotiation phase will end with a slightly different game but with the same solution, and after each player making only one offer. As we will see further, such confluence is not always the case.  
\end{example}

\subsection{Weakness of unconditional offers} 

The example above demonstrates the potential power of unconditional offers to solve normal form games. On the other hand, the version of the Prisoners' Dilemma game in Figure \ref{PD2}  demonstrate their weakness, showing that in preplay negotiation games where no conditional offers and no withdrawals are allowed the players may be unable to reach any Pareto optimal outcome by means of exchanging feasible preplay offers. Moreover, the expected value of the game, that a player can achieve by making an effective unconditional offer in such a preplay negotiations game, can be  worse than the original expected value of the game yielded by the maxmin strategy profile, for \emph{every} player. 

\subsection{The disadvantage of making the first unconditional offer}

Even when each of the players can start an effective negotiation ending with a solved game, the solution may essentially depend on who makes the first effective offer, as shown by the next example.  

\begin{example}[Making the first offer can be disadvantageous]
The reader can check that in the following game between $R$ and $C$ 
\[
\begin{game}{3}{2}
             & $C1$  & $C2$\\
$R1$   & $1,8$   &$10,4$\\
$R2$   & $4,10$   &$1,11$ \\
$R3$   & $4,0$   &$2,2$ \\ 
\end{game}\hspace*{\fill}%
\]
if the first offer is made by $R$ the preplay negotiation game ends with 
\[
\begin{game}{3}{2}
             & $C1$  & $C2$\\
$R1$   & $1,8$   &$6^{-},8^{+}$\\
$R2$   & $4,10$   &$-3^{-},15^{+}$ \\
$R3$   & $4,0$   &$-2^{-},6^{+}$ \\ 
\end{game}\hspace*{\fill}%
\]
where the only acceptable 
(surviving iterated elimination of strictly dominated strategies) outcome is $(R1,C2)$ yielding payoffs $(6^{-},8^{+})$, whereas if the first offer is made by $C$ the preplay negotiation game ends with 
\[
\begin{game}{3}{2}
             & $C1$  & $C2$\\
$R1$   & $4^{+},5^{-}$   &$9,5$\\
$R2$   & $4,10$   &$-3^{-},15^{+}$ \\
$R3$   & $4,0$   &$-2^{-},6^{+}$ \\ 
\end{game}\hspace*{\fill}%
\]
where the only acceptable outcome is again $(R1,C2)$, but now yielding payoffs $(9,5)$. Note that in both cases the disadvantaged player is the one who has made the first offer. 
\end{example}

The example above also indicates that, in the case under consideration, the greedy approach, where a player always makes the best effective offer he can, may not be his best strategy. Passing the turn to the other player -- that is, making a vacuous offer -- could be strategically more beneficial. On the other hand, if both players keep exchanging only vacuous offers or passing, then they will never improve their expected values of the starting game. Yet, one can check that any pair of strategies in the example above, whereby one of the player takes the initiative by making the first effective move with his best first offer and thereafter always responds with his currently best effective offers until possible and then passing, while the other player remains passive until that happens and thereafter keeps responding with her best offers until possible and then passing, is a subgame-perfect equilibrium strategy in the preplay negotiation phase for that game.

\paragraph{Stocktaking}
We have demonstrated that, on the one hand, by exchanging only unconditional and irrevocable offers players can often achieve mutually better outcomes of normal form games, but on the other hand their bargaining powers to achieve their \emph{best}  outcomes in such games can be substantially affected by the potential disadvantage of making the first effective offer in such games. 
Consequently, the strategy profile based on always making the currently best offer need not always be a Nash equilibrium.
We therefore believe that the analysis of PNG with unconditional offers warrants the use of equilibria that go beyond the IRA assumptions. Also, features such as using vacuous offers for signaling a future intention on the strategy to be played become essential. Such analysis should take into account equilibria generated both by forward induction-like reasoning, where past moves can be used to justify rational behavior in the future (see \cite{OR}). We leave the overall analysis of this case and the further investigation of the best negotiation strategies and the analysis of the effect of valuable time, to future work. 
It is not yet known precisely what additional conditions guarantee existence of pure strategy Nash equilibrium for preplay negotiaitons with unconditional offers. This question is left to future study.

\section{
Two-players preplay negotiation games with conditional offers}
\label{sec:2-player-cond}



%

In this section we allow the possibility of players to make conditional offers to each other and 
obtain results about the efficiency of the resulting negotiation process and its possible outcomes, under several optional assumptions. The content of this section extends our work in \cite{GT-LORI2013}.

Before analyzing some cases with additional optional assumptions, let us state a useful general result, also valid in the case of many players PNG. 
An extensive form game is said to have the {\bf One Deviation Property} (ODP) \cite[Lemma 98.2]{OR} if, in order to check that a strategy profile is a Nash equilibrium in (some subgame of) that game, it suffices to consider the possible profitable deviations of each player not amongst {\em all} of its strategies (in that subgame), but only amongst the ones differing from the considered profile in the first subsequent move.

\begin{lemma}
\label{ODP}
Every PNG has the One Deviation Property.
\end{lemma}

\begin{proof}
Let $\mathcal{E}=(
\N, \act, \hist, \mathit{turn},\{\Sigma_A,\Sigma_B\},  o, \mathcal{G}, \solc, \bfg, \bfu)$ be a PNG, and $\hist_f \subseteq \hist$ the set of finite histories in $\hist$. Let moreover  
$\mathcal{E}_{f}$ be the restriction of $\mathcal{E}$ to $\hist_f$, where the individual components are defined in the expected way. 
%
But $\mathcal{E}_{f}$ is a game of finite horizon, and by \cite[Lemma 98.2]{OR} it has the One Deviation Property. But by the fact that 
no disagreement is better for any player than any agreement, 
(Definition \ref{def:PNG}) then $\mathcal{E}$ has that property, too.\end{proof}

Furthermore, to analyze equilibrium strategies of PNG  we
consider so called 
{\bf stationary acceptance strategies} where players have a minimal acceptance threshold 
$d$ and a minimal passing threshold 
$d^{\prime}\geq d$ (both of which may vary among the players).

\subsection{Conditional offers with non-valuable time}

The value for a player of a history in a PNG is the value for the player of the current NFG associated with that history. When time is not valuable players assign the same value to the NFG associated with the current moment and the same game associated with any other moment in the future, which means that players can afford delaying or 
withdrawing offers at no extra cost. 

\begin{proposition}\label{prop:stationary}
Every SPE  
strategy profile of stationary acceptance strategies of a two-player PNG with non-valuable time is strongly efficient.

\end{proposition}

\begin{proof}
Suppose not. Let $d^{-}$ be a vector of expected values that is not the redistribution of a maximal outcome of the starting game, associated to some SPE 
strategy profile. Such strategy profile yields a history $h$ that ends with: 1) proposal of $d^{-}$;  2) acceptance of that proposal;  3) pass; 4) pass. 
Consider now some redistribution $d^{*}$ of a maximal 
outcome where both players get more than in $d^{-}$ and the  history $h$ with the the last four steps substituted by: 1) proposal of $d^{*}$; 2) acceptance of that proposal; 3) pass; 4) pass. By stationarity of strategies and the ODP, the player moving at step 1) is better off deviating from $d^{-}$ and instead proposing $d^{*}$: a  contradiction. 
\end{proof}


\medskip
The condition of stationarity of acceptance strategies is needed if we want to prevent SPE that lead to inefficiency. Indeed, if players were not adhering to stationary acceptance strategies there could be a {\em suboptimal} outcome, guaranteeing for both players expected values respectively of $d_A$ and $d_B$. To enforce that outcome it then suffices to design a strategy profile whereby off the equilibrium path player $A$ threatens player $B$ with a {\em stubborn} but maximal stationary acceptance strategy giving him less than $d_B$, while player $B$ threatens $A$ with an expected payoff of strictly less than $d_A$. So, if players are not obliged to be consistent in their acceptance policies, $d_A$ and $d_B$ can be the result of a subgame perfect equilibrium strategy. 

The example below provides a detailed instance of such games.

\begin{example}[Attaining inefficiency]\label{exa:inefficiency}

In what follows we say that a player `proposes a given outcome with a given payoff distribution' to mean that the player makes a conditional offer which, when accepted, would make that specific outcome, with that specific distribution of the payoffs, the unique (dominant strategy equilibrium) outcome in the solution of the transformed game. 
More generally, we say that a player "proposes a payoff distribution" to mean that the player  makes a conditional offer which, when accepted, would make that specific payoff distribution the vector of expected utilities of the players.

Consider the starting NFG on Figure \ref{fig:inefficiency}. As there are no dominant strategy equilibria, there are acceptable solution concepts assigning $2$ to each player.
\begin{figure}[htb]\hspace*{\fill}%
\begin{game}{2}{2}
      & $L$    & $R$\\
$U$   &$2,2$   &$4,3$\\
$D$   &$3,3$   &$2,2$
\end{game}\hspace*{\fill}
\caption{Attaining inefficient divisions}
\label{fig:inefficiency}
\end{figure}

We now will construct a strategy profile of the PNG starting from that game,  
that is a SPE strategy profile and attains an inefficient outcome: 

\begin{enumerate}
\item
At the root node player $A$ proposes outcome $(D,L)$ with payoffs $(3,3)$. 
\item
After such proposal player $B$ accepts. However, if $A$ had made a different offer (so, off the equilibrium path)  $B$ would reject and keep proposing outcome $(U,R)$ with  distribution  of $5$ for him and $2$ for $A$ and accepting (and passing on) maximal outcomes guaranteeing him at least $5$. $A$, on the other hand, would not have better option than proposing the same distribution ($5$ for $B$ and $2$ for her) and accepting only maximal outcomes guaranteeing her at least $2$. Notice that once they enter this subgame neither $A$ nor $B$ can profitably deviate from such distribution.
\item
 If, however, $B$ did not accept the $(3,3)$ deal then $A$ would keep proposing outcome $(U,R)$  with a redistribution of $(5,2)$ ($5$ for her, $2$ for him) and accepting at least that much.  Respectively, $B$ would also stick to the same distribution, accepting at least $2$. Again, no player can profitably deviate from this 
stationary strategy profile starting from $B$'s rejection.

\item
After player $B$ has accepted the deal $(3,3)$, then $A$ passes. If $A$ did not pass, player $B$ would go back to his $(2,5)$ redistribution threat. 

Likewise with the next round. That eventually leads to the inefficient outcome 
$(3,3)$. 
\end{enumerate}

It is easy to check that the strategy profile described above is a SPE. 
No player can at any point deviate profitably by proposing the outcome $(U,L)$ with dominating payoff distribution, e.g., $(3.5, 3.5)$ . 

\end{example}

\medskip
We first focus on PNG where the opt out option is not available, and introduce it as an additional feature later on.

\paragraph{Negotiations without 'opt out' moves}

In PNG with non-valuable time and without the possibility of opting out every  redistribution of a maximal outcome can be attained as a solution.

\begin{proposition}\label{prop:subgame1}
Let $\mathcal{E}$ be PNG with non-valuable time starting from a NFG $\mathcal{G}$ and 
\label{prop1} let $d=(x_A,x_B)$ be any redistribution of a maximal outcome of the starting 
NFG.
The following strategy profile $\sigma = (\sigma_{A},\sigma_{B})$ is a SPE: 

For each player $i \in \{A,B\}$:
\begin{itemize}
\itemsep = -2pt
\item if $i$ is the first player to move, he proposes a transformation of $\mathcal{G}$ where the vector of expected values in the transformed game is $d$;
\item when $i$ can make an offer and the previously made offer has not been accepted, he proposes a transformation of the current NFG where the vector of expected values in the transformed game is $d$; 
\item when $i$ can make an offer and the previously made offer has been accepted, he passes;
\item when $i$ has a pending offer of a suggested transformation where the vector of expected values in the transformed game is $d^{\prime}$, he accepts it if and only if $x^{\prime}_i \geq x_i$, and rejects it otherwise;
\item $i$ never withdraws any previously made offer; 
\item if $i$ can pass and the other player has just passed, he passes; 
\item if $i$ can pass and the opponent has not just passed, $i$ proposes $d$; 
\item if $i$ has just accepted a proposal he passes;
\end{itemize}
\end{proposition}

\begin{proof}
We have to show that there is no subgame where a player $i$ can profitably deviate from this strategy at its root. By Lemma \ref{ODP} it suffices to consider only first move deviations to the above described strategy. 

Suppose the player has a pending offer  that induces a transformation of the current NFG where the vector of expected values is $d^{*}$. If she accepts it then the outcome will be $d^*$, due to the definition of the strategy profile; if she rejects it, it will be the starting offer $d$. And she will accept if and only if she will get more from $d^{*}$ than from $d$. So the acceptance component is optimal. For the remaining cases, if player $i$ deviates from the prescribed strategy, due to the construction of the strategy and Lemma \ref{ODP}, the vector of payoffs associated to the outcome of $\mathcal{E}$ will be $d$ anyway.  
\end{proof}


\begin{corollary}
\label{cor2} The game associated to the outcome of a subgame perfect equilibrium strategy profile consisting of stationary acceptance strategies in a two-player PNG with non-valuable time is optimally solvable. 
\end{corollary}

In summary, our analysis of two-player PNG with non-valuable time shows that efficiency can be attained when conditional offers are allowed and stationary acceptance strategies are followed. Indeed, any redistribution of the vector of payoffs of a maximal outcome can  be  made the unique solution of the final NFG 
by such SPE strategies. However, non-stationary acceptance strategies may lead to inefficient equilibria,  as Example \ref{exa:inefficiency} clearly shows: there exist SPE 
strategy profiles  of a two-player PNG with conditional offers and non-valuable time where (i) offers are made that are not feasible, (ii) the vector of payoffs of the outcome it attains is not a redistribution of the vector of payoffs of a maximal outcome, i.e., it is not strongly efficient.

\paragraph{Negotiations with `opt out' moves.}

To address the issues related to possible inefficiency we consider the possibility for players to make an {\em opt out} move and unilaterally put an end to the negotiations, by making the currently accepted NFG the outcome of the whole PNG. 

\begin{proposition}
Let $\sigma$ be a subgame perfect equilibrium strategy profile of a PNG with  {\em opt out} move and let $h$ be the resulting history. Then $\sigma$ guarantees to all players at least as much as they had in the currently accepted NFG;
in particular, at least as much as in the original game.
\end{proposition}

\begin{proof}
Starting with the original, automatically accepted game, each currently accepted NFG must make each player better off than in the previous one;  otherwise opting out would be a profitable deviation. \end{proof}

By introducing the possibility of opting out, the set of subgame perfect equilibria reduces further. Strategies, such as the one described in Example \ref{exa:inefficiency}
demanding an unreasonably high reward or an unreasonably low one for the proponent, will not be equilibria anymore.  However, this option does not solve the problem of attaining inefficiency, as the comment to Proposition \ref{prop:stationary} 
still applies. It has, however, several advantages: first, the equilibrium strategies of the PNG will guarantee for both players at least the expected payoff of the starting NFG; and second, the threat of opting out gives the players the possibility of making a more effective use of unconditional offers. 


\medskip

To sum up, while SPE strategies in a two-player PNG can attain efficiency, some important issues are still remaining:

\begin{itemize}
\itemsep = -2pt
\item some SPE strategies, e.g., non-stationary acceptance strategies, are not strongly efficient.

\item players can keep making unfeasible moves as a part of a SPE strategy, i.e.,  there are forms of equilibria where some players strictly decrease their expected payoff with respect to the original game; 

\item even strongly efficient strategies do not always yield perfectly solved games, as there is no notion of {\em most fair} redistribution of the payoff vectors in the solution of the original game.

\end{itemize}

Thus, when time is of no value, even the possibility of making conditional offers does not  guarantee that fair and efficient outcomes are ever reached. 

\subsection{Conditional offers 
with valuable time}

We will show here that when time is of value the problems mentioned above can be at least partially solved. To impose value on time we introduce for each player $i$ a {\em payoff discounting  factor $\delta_i \in (0,1)$} applied at every round of the PNG {\em associated to offers that are made} to his payoffs.  
These factors measure the players' impatience, i.e., how much they value time, and  reduce the payoffs accordingly as time goes by. Thus, the players have no interest in delaying the negotiations by making redundant moves, sub-optimal or subsequently withdrawn offers.
The intuition now,  which we will justify further, is that  
\pcut{the only SPE strategy profiles for the preplay negotiation games with valuable time 
would consist of just 2 moves: 

\begin{enumerate}
\itemsep = -2pt
\item The first player to move makes his best conditional offer that the other player would ever accept (by adding for himself a small premium for the time saving). 
\item The other player accepts that conditional offer. 
\end{enumerate}

The reasoning behind this intuition is that,
} 
for the sake of time efficiency, in a SPE strategy profile:  

\begin{enumerate}
\item If a player intends to make an offer, she has never made any earlier offer that, if accepted, would give her a lesser value of the resulting game.  

\item If any player is ever going to accept a given offer (or any other offer which is at least as good for her) 
she should do it the first time when she receives such offer. 
\end{enumerate}


In analyzing PNG with valuable time we consider several cases, depending on whether 
withdrawals and opting out are allowed. 

\subsubsection{No withdrawals and no opting out}\label{NoW-NoOpt}

For technical reasons we impose some additional constraints:

\begin{itemize}
\itemsep = -2pt
\item every game associated with a history of a PNG does not have  \emph{in its solution} outcomes assigning negative utility to players.
NB: we do allow payoff vectors consisting of negative reals to be present in the game matrix, only we do not allow such vectors to be associated to outcomes in the solution. 
This constraint has several practical consequences:

\begin{itemize} 
\itemsep = -2pt
\item players' expected payoffs {\em decrease} in time, 
i.e., the discounting factor $\delta$ has always a negative effect on the expected payoff.
\item players can make offers that redistribute the payoff vectors associated with outcomes in the solution, leaving some nonnegative amount to each  player and some strictly positive amount to some. 
\end{itemize}

\item each player's expected payoff at a disagreement history is assumed $0$.
\end{itemize}


We will use the following notational conventions:

\begin{itemize}
\itemsep = -2pt
\item $(x,t)$ denotes the payoff vector $x$ at time $t$, where each component $x_i$ is discounted by $\delta_i^{t}$; $(x,t)_i$ is the payoff of player $i$ in the vector $x$ at $t$. 
\item  $\mathcal{G}_X$ will denote the set of all possible redistributions of payoffs of 
outcomes in a NFG $\mathcal{G}$ that assign nonnegative payoffs to all players. 
This set is compact, but generally not connected, as in the bargaining games of [4]. However, it is a finite union of compact and connected sets, and that will suffice to generalize the results from [4] that we need.
\end{itemize}

\medskip
\label{obs:OR-assumptions}
The following properties of every 2-person PNG with valuable time starting from a given 
NFG $\mathcal{G}$ are the four fundamental assumptions of 
the  bargaining model in \cite{rub82} and \cite[p.122]{OR}.

\begin{enumerate}
\itemsep = 1pt
\item For each $x,y \in \mathcal{G}_X$ such that $x \neq y$, if $(x,0)_i = (y,0)_i$ then $(x,0)_{-i} \neq (y,0)_{-i}$. This holds because the set $\mathcal{G}_X$ is made by payoff vectors and subtracting some payoff to a player means adding it to the other.

\item $(b^{i},1)_{-i} = (b^{i},0)_{-i} = (D)_{-i}$, where $b^{i}$ is the highest payoff that $i$ obtains in $\mathcal{G}_X$ and $(D)_{-i}$ the payoff for $-i$ in any disagreement history. As $b^{i}$ is the best agreement for player $i$ it is also the worst one for player $-i$. 

\item If $x$ is Pareto optimal amongst the payoff vectors in $\mathcal{G}_X$ then, 
by definition of $\mathcal{G}_X$, there is no $y$ with $(x,0)_i \geq (y,0)_i$ for each $i\in N$. Moreover, $x$ is a redistribution of a maximal outcome in $\mathcal{G}$. 

\item There is a unique pair $(x^{*},y^{*})$ with $x^{*},y^{*} \in \mathcal{G}_X$ 
such that $(x^{*},1)_A = (y^{*},0)_A \mbox{ and } (y^{*},1)_B=(x^{*},0)_B$ and both $x^{*},y^{*}$ are Pareto optimal  amongst the payoff vectors in $\mathcal{G}_X$.
\end{enumerate}

 The first 3 statements above are quite straightforward. To see the last one,  let $x^{*} = (x_{A}^{*},x_{B}^{*})$ and  $y^{*} = (y_{A}^{*},y_{B}^{*})$ and let the sum of the payoffs in any maximal outcome in  $\mathcal{G}$ be $d$.  Then $(x_{A}^{*},x_{B}^{*},y_{A}^{*},y_{B}^{*})$ is the unique solution of the following, clearly consistent and determined system of equations: \\
$y_{A} = \delta_{A} x_{A}$, 
$x_{B} = \delta_{B} y_{B}$, 
$x_{A} + x_{B} = d$, 
$y_{A} + y_{B} = d$.   

The solution (see also  \cite{OR}) is: 
\[x_{A} = d\frac{1- \delta_{B}}{1- \delta_{A} \delta_{B}}; \ 
y_{A} =  \delta_{A}d \frac{1- \delta_{B}}{1- \delta_{A} \delta_{B}}\]

\[x_{B} =  \delta_{B} d\frac{1- \delta_{A}}{1- \delta_{A} \delta_{B}}; \ 
y_{B} = d\frac{1- \delta_{A}}{1- \delta_{A} \delta_{B}}.\]

\vcut{
We can now view the preplay negotiation as a bargaining process on how to play the starting normal form game and can adapt the results from \cite{OR}. Indeed, it can now be shown that the four fundamental assumptions of Rubinstein's perfect equilibrium solution of the bargaining problems in  \cite{rub82}, see also  Osborne and Rubinstein's bargaining model \cite[p.122]{OR} are satisfied for the resulting bargaining game and, in particular, there is a unique pair $(x^{*},y^{*})$ with $x^{*},y^{*} \in \mathcal{G}_X$ such that $u_A(x^{*},1) = u_A(y^{*},0) \mbox{ and } u_B(y^{*},1)=u_B(x^{*},0)$ and both $x^{*},y^{*}$ are Pareto optimal  amongst the payoff vectors in $\mathcal{G}_X$.
 Consequently, we obtain the following: 
}

\paragraph{Relation with bargaining games}
In the rest of the section we will explicitly view preplay negotiation as a bargaining process on how to play the starting normal form game. Using our observations and assumptions, we can adapt the results from \cite{OR} to show that when time is valuable not only all equilibria {consisting of stationary acceptance strategies} attain efficiency but they also do so by redistributing the payoff vector in relation to players' impatience. {Stationary acceptance strategies will be needed to focus only on the maximal connected subspace of the set $\mathcal{G}_X$. 
\vcut{
To say it with a slogan, while in \cite{OR} efficiency and fairness can be obtained in scenarios that resemble the division of a cake, in our setting we prove similar results for a set of cakes, of possibly different size. 
}
We extend the efficiency and fairness results obtained in \cite{OR}  for
bargaining games of the type of `division of a cake' to 
somewhat more general bargaining games of the type where players have to
choose a cake from a set of cakes, of possibly different sizes and divide it.
Our claim, in a nutshell, is that, when players employ stationary acceptance strategies, they immediately choose the largest cake and then bargain on how to divide it.

\cut{
Before doing this we say that making the first proposal in a PNG with valuable time is {\bf convenient} for player $A$ (resp. $B$) if the starting game $\mathcal{G}$ yields a vector of expected values $(c_A,c_B)$ with $c_A \leq \delta^{2}_A x^*_A$ (resp. $c_B \leq \delta^{2}_B y^*_A$). It is {\bf strictly convenient} for player $A$ (resp. $B$) if the previous inequality is strict. As we will see from the next result this definition precisely encodes when for a player is worth entering the negotiation game by making a first proposal, rather than passing directly instead.
}

First, recall that in our framework time passes as new offers are made. So, from a technical point if the PNG start with a game that is already perfectly solved, the player moving first will not be punished by passing immediately.

Then, without restriction of the generality of our analysis, we can assume a {\em unique} discounting factor for both players. Indeed, the discount factor of e.g., player $A$ can be made equal to that of $B$ while preserving the relative preferences of $A$ on the set of outcomes by suitably re-scaling the payoffs of $A$ in the input NFG, and therefore the expected value for $A$ of that game; for technical details see  \cite[p.119]{OR} following an idea of  Fishburn and Rubinstein quoted there.

Now we are ready to state the main result for this case: 

\begin{theorem}
\label{prop:unique}
Let $(x^{*},y^{*})$ be the unique pair of payoff vectors defined above. 
{Then, in a PNG with valuable time starting from a NFG $\mathcal{G}$ with a unique discounting factor $\delta$ for both players, the strategy of player $A$ in every subgame perfect equilibrium consisting of stationary acceptance strategies 
satisfies the following  (to obtain the strategy for $B$ simply swap $x^{*}$ and $y^{*}$):}

\begin{itemize}
\itemsep = -2pt
\item if $A$ is the first player to move, then she 'proposes' outcome $x^{*}$, i.e., makes a conditional offer that, if accepted, would update the game into one with a dominant strategy equilibrium yielding the Pareto maximal outcome $x^{*}$ as payoff vector;
\item when $A$ has a pending offer $y^{\prime}$, she accepts it if and only if the payoff she gets in $y^{\prime}$ is at least as much as in $y^{*}$;
\item if $A$ can pass, she passes if and only if the expected value associated to the proposed game $y^{\prime}_A$ is at least $y^{*}_A$;  otherwise she proposes $x^{*}$.
\end{itemize}
\end{theorem}


\begin{proof}
It is easy to check, using the ODP, that no player can improve at any history of the game by deviating from this strategy. Consider for instance the case when player $A$ at time $t$ can choose whether to pass or not on the proposal of a distribution $z$ 
on which player $B$ has already passed. If $A$ passes then the payoff vector will be $(z,t)$;  if not, it will be $(x^{*},t+1)$ . Obviously $(z,t)_A \geq (x^{*},t+1)_A$ if and only if $(z,0)_A \geq (x^{*},1)_A = (y^*,0)_A $, so the a,cceptance rule is optimal. The reasoning for the other cases is similar.

To prove the claim we use 
a variant of the argument  given in \cite{OR} for bargaining games, summarized as follows.  We first show [Step 1] that the best SPE payoff for player $A$ in any subgame $\mathcal{G^{\prime}}_A$  starting with her proposal and where $\mathcal{G}^{\prime}$ is the currently accepted game --- let us denote it by $M_A(G^{\prime}_A)$ --- yields the same utility as the worst one --- $m_A(\mathcal{G^{\prime}}_A)$ --- which, in turn, is the payoff of $A$ at $x^{*}$. The argument for $B$ is symmetric. 
Then we show [Step 2] that in every SPE the initial proposal is $x^{*}$, which is immediately accepted by the other player, followed by each player passing. Finally, we show [Step 3] that the acceptance and the passing conditions given are shared by every SPE strategy profile.

[Step 1] WLOG let $A$ be the player moving first and call $\mathcal{G^{\prime}}_A$ each subgame of the PNG beginning with a proposal by player $A$ and where $\mathcal{G^{\prime}}$ is the currently accepted NFG at its root ($\mathcal{G}_A$ is the game itself). Analogously let us call $\mathcal{G^{\prime}}_B$ each subgame of the PNG beginning with a proposal by player $B$. For each player $i$ let $M_i(\mathcal{G^{\prime}}_i)$ be the best SPE outcome that player $i$ can get from $\mathcal{G^{\prime}}_i$, i.e., $M_i(\mathcal{G^{\prime}}_i) = sup\{\delta^{t}x_i \mid$ there is a SPE of  $\mathcal{G^{\prime}}_i$ {consisting of stationary acceptance strategies} with value $ (x,t)_i \}$. Let $m_i(\mathcal{G^{\prime}}_i)$ be the corresponding infimum. Recall that $b^{i}$ is the highest payoff that $i$ obtains in $\mathcal{G}_X$. Hereafter we write $b^{i}_j$ instead of $(b^{i},0)_j$ for $i,j\in N$. By our assumptions the observations above, $b^{A}_B = b^{B}_A = 0$.   

We can now show that for each $\mathcal{G^{\prime}}$, $M_A(\mathcal{G^{\prime}_A})=m_A(\mathcal{G^{\prime}}_A)=x^{*}_A$ and $M_B(\mathcal{G^{\prime}}_B)=m_B(\mathcal{G^{\prime}}_B)=y^*_B$.
We first show that $m_B(\mathcal{G^{\prime}}_B) \geq b^{B}_B - \delta M_A(\mathcal{G^{\prime}}_A)$. 
Therefore, if player $A$ rejects a proposal of player $B$ in the first period of $\mathcal{G^{\prime}}_B$ then she cannot get more than $\delta M_A(\mathcal{G^{\prime}}_A)$. This means that in any SPE of $\mathcal{G^{\prime}}_B$ she must accept any proposal giving her more than $\delta M_A(\mathcal{G^{\prime}}_A)$ (otherwise she could be at least as well off by rejecting it). Thus what is left for player $B$ is no less than $b^{B}_B - \delta M_A(\mathcal{G^{\prime}}_A)$ in any SPE of $\mathcal{G^{\prime}}_B$. 

It is easy to see that $M_A(\mathcal{G^{\prime}}_A) \leq b^{A}_A -  \delta m_B(\mathcal{G^{\prime}}_B)$, because player $A$ cannot get more than her best agreement minus what player $B$ could guarantee with a rejection. That is, player $A$ needs to pay $B$ with the difference between her ideal (appropriately discounted) payoff and what $B$ could guarantee alone. We can show now that $M_A(\mathcal{G^{\prime}}_A) = x^{*}_A$. That $M_A(\mathcal{G^{\prime}}_A) \geq x^{*}_A$ is easily observed from the properties satisfied by every SPE and the fact that each $\mathcal{G}^{\prime}$ is a transformation of $\mathcal{G}$ by conditional offers. To show that $M_A(\mathcal{G^{\prime}}_A) \leq x^{*}_A$ we argue the following.
We know that $\delta b^{A}_B = 0$. We also know that $\delta (b^{B}_B - \delta b^{A}_A) >0 = b^{A}_B= b^{B}_B - b^{A}_A$. In turn we have that $b^{A}_A > b^{A}_A - (\delta ( b^B_B - \delta b^{A}_A))$. By the previous observations we can conclude that $M_A(\mathcal{G^{\prime}}_A) \leq b^{A}_A - (\delta ( b^B_B - \delta M_A(\mathcal{G^{\prime}}_A)))$. But, {by a similar argument to that in the proof of Proposition \ref{prop:stationary}, $M_A$ is obtained from a strongly efficient SPE}. So, as the set of maximal outcomes in $\mathcal{G}_X$ is compact and connected, it also follows that there exists $U_A \in [M_A(\mathcal{G^{\prime}}_A), b^A_A)$ such that $U_A =  b^{A}_A - (\delta ( b^B_B - \delta U_A))$. Now if $M_A (\mathcal{G^{\prime}}_A) > x^{*}_A$ then $U_A \neq x^{*}_A$. Then, taking any pair of efficient agreements $(a^{*},b^{*})$ such that $a^{*}_A= U_A$ and $b^{*}_A = \delta U_A$ 
we have obtained a pair of efficient agreements contradicting Property \ref{obs:OR-assumptions} (4). Similar reasoning shows that $m_A(\mathcal{G^{\prime}}_A)=x^{*}_A, M_B(\mathcal{G^{\prime}}_B) = y^{*}_B$ and finally $m_B(\mathcal{G^{\prime}}_B)=y^*_B$. 

\medskip
[Step 2] Step 1
implies that if $A$ is the first player to move, she starts by proposing $x^{*}$ which is immediately accepted. Likewise for player $B$.

\medskip
[Step 3] Step 1 and 2 imply that every SPE shares the same acceptance and passing condition. Consider first the acceptance condition. If $B$ rejects an offer in $\mathcal{G^{\prime}}_A$ we go to $\mathcal{G^{\prime}}_B$ where, by what was observed before, he gets $y^{*}_B$. But $y^{*}_B = \delta x^{*}_B$ so every proposal giving him in $\mathcal{G^{\prime}}_A$ at least $x^{*}_B$ should be accepted, otherwise rejected.
Putting everything together we have that player $B$ must accept any proposal giving him exactly $x^{*}_B$. Similar reasoning applies for the passing condition and for player $A$.\end{proof}

One important consequence of Theorem \ref{prop:unique} is that every SPE
strategy profile, {consisting of stationary acceptance strategies}, of a two-player PNG with valuable time and with $N=\{A,B\}$ starting from $\mathcal{G}$ and with $A$ (resp. $B$) first player to move induces a play $h$ of length $4$ and of value for player $A$ of $x^{*}_A$ while for player $B$ of $\delta y^{*}_B$ (resp. $(y^{*}_A,\delta x^{*}_B)$ if $B$ moves first). 


\medskip

To summarize, when time is valuable and players' value of time (impatience) is measured by a vector of discount factors $\delta$ and no withdrawals and opting out are allowed, the SPEs {following stationary acceptance strategies} are essentially unique, efficient and redistribute a maximal payoff vector in a {\em fair} way, depending on players' impatience, viz. in each SPE play, players agree as soon as possible and divide (almost) evenly any of the maximal outcomes in the game. Thus, introducing value of time solves both problems of efficiency and fairness at once. 
\section{Related work and comparisons}
\label{sec:related}

The present study has a rich pre-history and we do not purport to provide a comprehensive citation of all related previous work and literature here, but will only mention various links with earlier studies and then will discuss in more detail and compare with the most relevant recent work. 

\subsection{Related topics and relevant early references} 

Here is a selection of related topics and relevant earlier references:

\medskip
\indent $\triangleright$ To begin with, preplay offers technically fall broadly in the scope of \emph{externalities}. There is abundant literature on these, of which we only mention some of the early works:  \cite{meade}, \cite{maskin}, \cite{var1}, 
More specifically, preplay offers can be regarded as a special type of so called in cooperative game theory \emph{side payments}. 

\medskip
\indent $\triangleright$ Coase theorem, \cite{coase} describes how efficiency of an allocation of goods or simply an outcome can be obtained in presence of externalities, i.e. when actors' possible decisions affect positively or negatively the payoffs of the other actors involved. The claim, which is usually provided in a rather informal fashion, states that if there are no transaction costs and it is possible to bargain on the effect of the externalites, the process will lead to an efficient outcome regardless of the initial allocation of property rights, i.e. regardless of who is endowed with the capacity of performing the action in question. 

\medskip
\indent $\triangleright$ 
\cite{ros1} proposes one of the earliest models of preplay negotiations, where `players successively commit themselves irrevocably, according to a specified exogenous ordering, to coalitional strategies conditionally on the rest of the players in the coalition agreeing to play their parts of the coalitional strategy'. He defines a special solution concept, the induced outcome, and provides some sufficient conditions for its existence and uniqueness.

\medskip
\indent $\triangleright$ Several two-stage games with preplay communication have been studied in the literature. They seem to go back to \cite{gut1}  and \cite{gut2}.  \cite{kal1} studies preplay negotiation procedures 
as sequences of pre-defined length of ``preplays'', each being a joint strategy of all players. \cite{MatthewsP89} consider preplay communication in the context of two-person sealed-bid double auctions. 
\cite{DS91} consider a 2-stage game for implementing Lindahl's voluntary-exchange mechanism. In a series of papers, incl. \cite{farrell}, Farrell considers two-stage games, with preplay `cheap talk' followed by actual play, and discusses the role of preplay communication in ensuring Nash equilibrium profile in the actual play.  Also, \cite{Watson91} studies two-stage  2-person normal form games with preplay communication and \cite{AGV} study Stackelberg-solvable games with preplay communication.

\medskip
\indent $\triangleright$ Our preplay negotiation games are closely related to \emph{bargaining} games, \cite{rub82,OR90,OR}, \cite{Myerson}. 

\medskip
\indent $\triangleright$
Another related early work is \cite{var1} where he studies variations  of `compensatory mechanisms'  where, instead of making offers, players declare compensations for which they are prepared to play one or another strategy (in favour of another player who is willing to pay such compensation and makes a binding offer for it). Although the flavour of  such variation is somewhat different, technically it reduces to a type of games with preplay offers that we have considered here.   

\medskip
\indent $\triangleright$ \cite{FJK}, and more recently \cite{MT06}, consider the use of `agents' or 'mediators' playing on behalf of the players, and show how such mechanisms can be used to achieve more efficient outcomes in non-cooperative games. 

\medskip
\indent $\triangleright$ The idea of combining competition and cooperation in non-cooperative games has been considered often since the early times of game theory, and has later evolved in theories of \emph{co-opetition} by \cite{BN} and more recently \cite{cor1}. Related in spirit are some theories of coalitional rationality, see \cite{amb1}.


\subsection{Detailed comparison with most relevant recent work} 

To our knowledge, Jackson and Wilkie have been the first to explicitly study arbitrary transfer functions from one to another player in a normal form game. That work was preceded by earlier relevant literature mentioned above, such as \cite{gut1, DS91, var1, qin}, where only limited forms of payments were considered, such as payments proportional to the actions taken by the other players or only contingent on own actions. Jackson and Wilkie's framework bears substantial similarities with ours, as it studies a two-stage transformations on a normal form game where players announce transfers functions which update the initial normal form game and then play the updated game. Jackson and Wilkie study the subgame perfect equilibria of the two stage game and show under what conditions equilibria of the original game survive in the update game. They focus on the 2-player case, but they also extend their results to the N-player case. However, there are some essential conceptual and technical differences between this framework and our, which we describe and discuss below. 
 In \cite{JW05}:
 

\medskip
\indent $\triangleright$  Transfers from a player $A$ to a player $B$ are of the form (in our notation) $A \xlongrightarrow{\delta \slash  \sigma} B$ where $\sigma \in \prod_{i \in N} \Sigma_i$, $\delta \in \mathbb{R}^+$ and $\delta=0$ whenever $A=B$, i.e. players are allowed to make positive side payments to other players {\em that are conditional on the entire strategy profile played}, and not only on the recipient's  individual strategy, as in our framework. 
Technically, every unconditional offer from player $A$ to player $B$ can be simulated by a set of such transfers from $A$ to $B$. This is not the case for conditional offers, which would instead require a {\em set} of transfers from $B$ to $A$ as well, or the possibility for $\delta$ to be negative, i.e. the introduction of punishments. So, these two types of offers are generally incompatible. 
%
However, we see the main importance of this difference as conceptual, rather  than technical. We argue that preplay offers based \emph{only} on the opponents' choice of actions are more natural and realistic than those dependent also on own or other players' actions, because of creating more explicit and unambiguous incentives for the opponents. 
Indeed, if a player $A$ makes an offer contingent upon a certain strategy profile $\sigma$ and hence, inter alia, on her playing a certain action $\sigma_{A}$, then $A$ creates positive incentives for the other players to play $\sigma$, but a possibly \emph{negative incentive for herself} to play $\sigma_{A}$. After all, if all other players take the bait and play $\sigma$ then $A$'s objective is already achieved, so rationally she should play her best response to  $\sigma_{-A}$ in the transformed game. If that action is different from $\sigma_{A}$ then $A$ would moreover save the promised payments to the others because the strategy profile $\sigma$ was not actually played! 

\medskip
\indent $\triangleright$  Players announce their transfer functions {\em simultaneously}. 
 This is a reasonable choice in situations where, e.g., players only have the possibility for once-off communication exchange before the actual play, but it is not so in many others where they would rather negotiate on their choice of actions, as it trivializes the whole preplay negotiation phase which is central in our framework. 
In that sense, the framework of \cite{JW05} and our have essentially different scopes of applicability. 
\vcut{
Canceling an offer by another player is technically possible in such framework, by simply making payments at every outcome which cancel out, e.g. $A \xlongrightarrow{\delta \slash  \sigma} B$ and $B \xlongrightarrow{\delta \slash  \sigma} A$, but cannot be a deliberate response to an anticipated action. 
}

\medskip
\indent $\triangleright$  
The authors study strategies that can be {\em supported}, i.e. that they are subgame perfect equilibria of the two-stage game and Nash-equilibria \emph{of the original game} that also {\em survive} --- i.e. remain equilibria --- in the updated game. In particular, they focus on the (interesting) relation between the solo-payoff, i.e. the Nash equilibrium payoff that a player can guarantee by making offers, and the supportability of strategies. 
Jackson and Wilkie show two important results for the two-player case, the main bulk of their paper: (i) that every Nash equilibrium $x$ of the starting game survives if and only if it yields for every player $i$ a utility that is higher than the one given by $i$'s solo-payoff; and (ii) that a transfer function together with an outcome are supportable if and only if they yield  for every player $i$ a utility that is higher than the one given by $i$'s {\em minimal} solo-payoff, the solo-payoff obtained by making {\em minimal} offers. It is worth noticing that the definition of minimal offer they adopt is essentially the one we have adopted here: the minimal transfer function needed to change the game solution.
 \vcut{
\item The role of time is not considered and players cannot build upon the game obtained in the second stage, by for instance subsequently making further offers.
}

%
%

\medskip

Ellingsen and Paltseva generalize Jackson and Wilkie's work as follows:

\smallskip
\indent $\triangleright$  Transfers from a player $A$ to a player $B$ are again of the form $A \xlongrightarrow{\delta \slash  \sigma} B$ where  $\sigma \in \prod_{i \in N} \Sigma_i$, but now $\delta \in \mathbb{R}$ and $\delta=0$ whenever $A=B$, i.e. players are allowed to propose both rewards and punishments contingent upon entire strategy profiles. This boils down to players not only making offers but also proposing contracts to the other players to sign or reject.

\smallskip
\indent $\triangleright$  The game played is composed of three stages: (i) the one in which players propose contracts, (ii) the one in which players decide whether to sign a contract, (iii) the one in which players play the game updated by the signed contract.

\smallskip
\indent $\triangleright$  Contracts are proposed on mix strategies, and non-deterministic contracts are considered, i.e. it is possible to make randomize offers. 

While in \cite{JW05} each player $A$ specifies the nonnegative transfer to the other players for each pure strategy profile $\sigma$, in \cite{EP11} each player specifies a (possibly negative) transfer to the other players  for each (possibly mixed) strategy profile $\sigma$ and, at the same time, specifies a signing decision for each contract of the other players. Ellingsen and Paltseva show that their more general contracting game always has efficient equilibria. In particular they show that all the efficient outcomes guaranteeing to each player at least as much as the worst Nash-equilibrium payoff in the original game can be attained in some equilibrium.

\cite{yamada-contracts} considers variants of the games  in \cite{JW05} where one player moves before the other and the move of the second ends the preplay phase, showing a clear advantage of the latter player in improving is own payoff. In particular, Yamada shows that: 
\begin{itemize} 
\itemsep = -2pt
\item the second player can always increase his original payoff, i.e. the payoff he gets in the starting game, in every surviving Nash equilibrium 

\item every surviving Nash equilibrium that is also maximally Pareto optimal gives the second player at least his original payoff
\end{itemize}

Clearly Yamada's framework is a step closer to ours than Jackson and Wilkie's. However the games analyzed there are a rather restricted sort of Stackelberg games, where the second player behaves like a dictator: not only can he best respond to the first player, but he can unilaterally decide that the game will end with him improving his original payoff. 

All in all, the message conveyed by this stream of contributions is that efficiency can be reached if the structure of players' offers is  complex enough. 
On the one hand Jackson and Wilkie show that promises are not enough to attain efficient outcomes, while Ellingsen and Paltseva show that contracting is. Possibly only Yamada's framework acknowledges that the structure of the game might influence the preplay phase. Our results lie on a rather different axis, as we restrict the type of offers to ones that only commit the proposer, not the recipient, and focus on the effects that additional factors in the preplay negotiation game, e.g. value of time, conditional offers and withdrawals, have on attaining outcomes with desirable properties, such as efficiency and fairness. We also discuss how equilbirium strategies themselves display desirable properties, i.e. being efficient negotiation strategies.

\section{Further agenda and concluding remarks}
\label{sec:conc}

The main purpose of the present paper is to initiate a systematic study of preplay negotiations in non-cooperative games, and to outline a broad and long-term research agenda for that study. We have indicated a number of conceptual and technical problems and have only sketched some results, but still much work needs to be done.
In particular, we identify two natural and important directions of current and future extensions of our framework:

\medskip
\indent  
\textbf{Coalitional offers.} The analysis of $N$-player normal form games with preplay negotiations phase, for $N > 2$, is much more complicated than the 2-players case. To begin with, the benefit for a player $A$ of player $B$ playing a strategy induced by an offer from $A$ to $B$ crucially depend on the strategies that the remaining players choose to play, so an offer from a player to another player does not have the clear effect that it has in the 2-player case. Thus, 
a player may have to make a  \emph{collective offer} to several (possibly {all}) other players in order to orchestrate their plays in the best possible for him way. Furthermore, a player may be able to benefit in different ways by making offers for side payments to different players or groups of players, and the accumulated benefit from these different offers may or may not be worth the total price paid for it. Lastly, when all players make their offers pursuing their individual interests only, the total effect may be completely unpredictable, or even detrimental for all players. It is therefore natural that groups of players get to collaborate in coordinating their offers. Thus, a coalitional behaviour naturally emerges here, and the preplay negotiation phase incorporates playing a coalitional game to determine the partition of all players into coalitions that will coordinate their offers in the negotiation phase. However, we emphasize again that the transformed normal form game played after the preplay negotiation phase should remain a non-cooperative game where every player eventually plays for himself.  

\medskip
\indent 
\textbf{Inter-play offers in extensive form games.} The problem of underperformance is not limited to normal form games, where players cannot observe the outcome of the opponents' actions during the play. It also arises in some extensive form games, such as the Centipede game, where the Backward Induction strategy profile can prescribe to players an utterly inefficient solution. The idea of preplay offers of payments to other players can be applied quite effectively in extensive form games by means of \emph{inter-play offers}, where, before every move of a player, the other player(s) can make him individual or coalitional offers conditional on his forthcoming move. The players from both sides can consider these offers through some commonly accepted solution concept, e.g. Backward Induction, which would provide current values for each player of every  subgame arising after the possible moves of $A$. 

In conclusion, the focal problems of the study initiated here are to: 

\begin{itemize}
\itemsep = -2pt
\item  analyze the game-theoretic effects of preplay/interplay offers for payments between individual players and coalitions in strategic and extensive form games, with complete and incomplete information;
\item develop the theory of preplay negotiations and, in particular, the concept of efficient negotiations under various assumptions considered here;  
\item analyze the optimality and efficiency of the solutions that can be achieved in preplay negotiation games; 
\item expand the study into a systematic theory of cooperation through negotiations in non-cooperative games. 
\item apply the developed theory and the obtained results both descriptively and prescriptively to real-life scenarios where our framework applies. 
\end{itemize}

\bibliographystyle{alpha}
\bibliography{Preplay-Bibliography}

\end{document}